\documentclass[12pt,a4paper]{article}

\usepackage{amsmath,amssymb}
\usepackage{amsfonts}
\usepackage{amsthm}
\usepackage{dsfont}
\usepackage{bbm}
\usepackage{a4wide}
\usepackage{graphicx}

\newcommand{\ud}{\mathrm{d}}
\newcommand{\ui}{\mathrm{i}}
\newcommand{\ue}{\mathrm{e}}

\newcommand{\R}{\mathds{R}}
\newcommand{\Z}{\mathds{Z}}
\newcommand{\N}{\mathds{N}}
\newcommand{\C}{\mathds{C}}
\newcommand{\Q}{\mathds{Q}}

\newcommand{\Prob}{\mathds{P}}
\newcommand{\E}{\mathds{E}}

\newcommand{\cU}{{\mathcal U}}

\newcommand{\ba}{\mathbf{a}}
\newcommand{\bL}{\mathbf{L}}

\newcommand{\bsigma}{\mathbf{\sigma}}

  \newcommand{\erf}{\operatorname{erf}}

\providecommand{\norm}[1]{\lVert#1\rVert}

\providecommand{\biggnorm}[1]{\bigg\lVert#1\bigg\rVert}
\providecommand{\abs}[1]{\lvert#1\rvert}
\providecommand{\bigabs}[1]{\big\lvert#1\big\rvert}
\providecommand{\biggabs}[1]{\bigg\lvert#1\bigg\rvert}

\newcommand{\la}{\langle}
\newcommand{\ra}{\rangle}

\newcommand{\bx}{\mathbf{x}}

\newcommand{\be}{\boldsymbol{e}}

\newcommand{\bA}{\mathbf{A}}

\newcommand{\bpm}{\begin{pmatrix}}
\newcommand{\epm}{\end{pmatrix}}

\newcommand{\veps}{\varepsilon}

\newtheorem{thm}{Theorem}
\newtheorem{lem}{Lemma}

\newtheorem{cor}{Corollary}
\newtheorem{Def}{Definition}
\newtheorem{cond}{Condition}

\begin{document}
\title{Entropy of eigenfunctions on quantum  graphs}

\author{Lionel Kameni and Roman Schubert\thanks{ \tt roman.schubert@bristol.ac.uk} \\  School of Mathematics \\ 
University of Bristol }

\maketitle
\abstract{We consider families of finite quantum graphs of increasing size and we are interested in how  eigenfunctions are distributed over the graph. 
As a measure for the distribution of an eigenfunction on a graph we  introduce the entropy, it has the property that  a large value of the 
entropy of an eigenfunction implies that  it cannot be 
localised on a small set on the graph. We then derive lower bounds for the entropy of eigenfunctions which depend on the topology of the graph and the 
boundary conditions at the vertices. The optimal bounds are obtained for expanders with large girth, the bounds are similar to the ones obtained by 
Anantharaman et.al.  for eigenfunctions on manifolds of negative curvature, and are based on the entropic uncertainty principle.  For comparison we  compute as well the 
average behaviour of entropies on Neumann star graphs, where the entropies are much smaller. Finally we compare  our lower bounds with numerical  results for regular graphs and star 
graphs with different boundary conditions.  }

\numberwithin{equation}{section}

\section{Introduction}

Differential operators on metric graphs have an interesting and rich spectral theory and 
can serve as  model systems for the study of questions from spectral geometry, 
quantum chaos and mathematical physics, see \cite{GnuSmi06,BerKuch13}. In this paper we will focus on 
the Laplacian on a metric graph with suitable boundary conditions at the vertices and we are interested in
the  distribution of the eigenfunctions and how this distribution depends on the topology of the graph and on the boundary conditions. 

One of the main open questions in this area is if there holds an analogue of the quantum ergodicity theorem. 
This theorem states that if $(M,g)$ is a compact Riemannian manifold whose geodesic flow is 
ergodic, then almost all eigenfunctions of the  Laplace Beltrami operator become equidistributed in the high energy limit. 
So far for graphs only partial results are available, quantum star graphs have been shown 
not to be  quantum ergodic, \cite{BerKeaWin04}, and in \cite{BerKeaSmi07} a special class of graphs 
was constructed which are quantum ergodic. In \cite{GnuKeaPio08, GnuKeaPio10} a much more general approach towards a the study of the statistical distribution of 
eigenfunctions and quantum ergodicity was developed, but the methods are not yet rigorous. It is important to note that for quantum graphs we 
expect quantum ergodicity to hold only in the limit of large graphs, which corresponds to the semiclassical limit. For graphs of fixed size a variety 
of limit measures can occur, see \cite{SchaKot03}, and they have been classified recently \cite{CdV14}.  Quantum graphs fall into the general class of systems 
with so called ray-splitting, and a general quantum ergodicity theorem for manifolds with ray splitting has been recently derived in \cite{JakSafStroh13}, but the results  
do not apply immediately to quantum graphs. 
  For the analogous problem on discrete graphs, with the 
eigenfunctions of the discrete Laplacian, a quantum ergodicity theorem for $d$-regular expanding graphs was established recently 
\cite{AnadMas13}. In this case  equidistribution  emerges as well only when the graph size tends to infinity. 

In this paper we will not study quantum ergodicity directly, but we will concentrate instead on the entropy of eigenfunctions and 
derive lower bounds in terms of geometric properties of the graphs. These estimates are inspired by  analogous  results on eigenfunctions 
on Riemannian manifolds by Anantharaman et.al., \cite{Ana08,AnaNon07}, and quantum maps in \cite{AnaNon07b,Gut10}.  We make in particular 
heavy use of  the entropic uncertainty principle,  as in  \cite{AnaNon07}. 
Lower bounds on the entropy imply constraints on how localised a limit measure of eigenfunctions can be, in particular a positive entropy 
excludes measures which are concentrated on a finite number of periodic orbits, so called strong scars. 

\section{Background and main results}

We will consider finite simple graphs $G=(V,E)$ with a vertex set $V$ and edge set $E$. We will denote the number of vertices by $\abs{G}$ and the number 
of edges by $\abs{E}$. The vertices are labeled by numbers $i\in \{1, 2, \cdots ,\abs{G}\}$ and any edge $e\in E$ can be labeled by the pair $(i,j)$
of vertices it connects, i.e, $e=(i,j)$. We will consider only undirected graphs, i.e., $(i,j)=(j,i)$ and that the graph is simple means that there are no multiple edges between any two vertices and no loops. The topology of the graph is encoded in the adjacency matrix 
$A=(a_{ij})$ which is a symmetric $\abs{G}\times \abs{G}$ matrix defined as 
\begin{equation}
a_{ij}=\begin{cases} 1 & (i,j)\in E\\ 0 & (i,j)\notin E\end{cases}\,\, .
\end{equation}

To each edge $e=(i,j)$ we give a length $L_e>0$ and we will 
identify the edge $e$ with the interval $[0,L_e]$ of length $L_e$. On $e=(i,j)$ we will use two coordinate systems, $x_{ij}\in [0,L_e]$ is defined by 
$x_{ij}=0$ denoting vertex $i$ and $x_{ij}=L_e$ vertex $j$. Then we have  $x_{ij}=L_e-x_{ji}$. The choice of coordinates introduces an orientation, and we will 
call an oriented edge a bond and denote it by $b=[i,j]$, then the reversed bond is $\hat b=[j,i]\neq [i,j]$.  We will denote the number of bonds by  $B=2\abs{E}$. 
A function on the graph is then a collection of $\abs{E}$ functions, 
one on each edge, $f_e : [0,L_e]\to \C$,  and the Laplace operator acts on each edge as the second order derivative, 
\begin{equation}
\Delta f_e=f_e''\,\, . 
\end{equation}
Hence an eigenfunction with 
eigenvalue $k^2$ is on edge $e=(i,j)$ of the form 
\begin{equation}\label{eq:EF}
f_e=a_{[i,j]}\ue^{\ui k x_{ij}}+a_{[j,i]}\ue^{\ui k x_{ji}}\,\, .
\end{equation}
In order that the eigenvalue problem is well defined we have to impose suitable boundary conditions at the vertices where several edges meet.  
These lead to unitary scattering matrices $\sigma^{(i)}$ at each vertex $i\in V$, see \cite{KosSchr99,Kuch04}.  If the vertex $i$ has degree $d_i$, then $\sigma^{(i)}$ is an 
$d_i\times d_i$ matrix, and the 
the boundary conditions for the eigenfunctions become 
\begin{equation}\label{eq:a-cond}
a_{[i,j]}=\sum_{j'\sim i} \sigma^{(i)}_{[ij],[j'i]}\ue^{\ui k L_{j'i}}  a_{[j',i]}\,\, ,
\end{equation}
where $i\sim j'$ means $j'$ has to be adjacent to $i$. The matrix $\sigma^{(i)}$ describes how incoming waves with wavenumber $k$ are scattered at the vertex $i$ onto outgoing waves, in general 
the matrix can depend on $k$, but in this paper we will restrict ourselves to two types of boundary conditions which lead to $k$-independent local $S$-matrices:

\begin{itemize} 
\item A function $f$ satisfies \emph{Neumann boundary conditions}, if at each vertex the function is continuous and the sum of the normal derivatives  at each vertex is zero.  For these the $S$-matrix at a vertex $i$ with degree $d_i$ reads  
\begin{equation}\label{eq:Neumann_S}
\sigma^{(i)}_{[ij],[j'i]} =\begin{cases} \frac{2}{d_i}-1 & j=j'\\  \frac{2}{d_i}  & j\neq j'\end{cases}\,\, .
\end{equation}
Notice that for Neumann boundary condition backscattering becomes dominant for large degree, in particular if $d_i\to\infty$ then $\sigma^{(i)}\to I$.  
\item \emph{Equi-transmitting boundary conditions}. These have been introduced in \cite{HarSmiWin07}, and they are characterised by the property that 
\begin{equation}
\abs{\sigma^{(i)}_{[ij],[j'i]}}^2=\begin{cases} 0 & j=j'\\ \frac{1}{d_i-1} & j\neq j'\end{cases}\,\, .
\end{equation}
With these boundary conditions backscattering is forbidden and an incoming wave is totally transmitted with equal probabilities 
to all outgoing bonds.  These boundary conditions do not exist for arbitrary degree $d_i$, the degree has at least to be even, and in this paper 
we will stick to the case that $d=p+1$, where $p$ is prime, and then we can chose 
\begin{equation}\label{eq:equi_trans_S}
\sigma:=\frac{1}{\sqrt{p}}\bpm 0 & 1 &\cdots & 1\\ 1 & &  &\\ \vdots & & C & \\ 1 & & & \epm
\end{equation}
with $C=(\chi(i-j))$ and $\chi(k)$ being the Legendre symbol
\begin{equation}\label{eq:Leg}
\chi(k)=\bigg(\frac{k}{p}\bigg):=\begin{cases} 0 & k=0 \mod p\\ 1 & k=\text{square}\mod p\\ -1 & k=\text{not square}\mod p\end{cases}\,\, .
\end{equation}
\end{itemize}

\begin{Def}\label{def:quantum_graph} A quantum graph $\hat G=(G, \bL, \bsigma)$ is a graph $G=(V,E)$ with a length $L_e>0$ assigned to each edge $e\in E$ and a unitary $d_i\times d_i$  matrix 
$\sigma^{(i)}$ assigned to each vertex $i\in V$.  The length are collected in the vector $\bL$ and the scattering matrices in the set $\bsigma=\{\sigma^{(i)}\, ;\, i\in V\}$.

To a quantum graph 
we associate its total scattering matrix $\cU_{\hat G}(k)$ which is a $B\times B$, $B=2\abs E$, unitary matrix with elements 
\begin{equation}
u_{[ij][kl]}= \delta_{jk} \sigma^{(j)}_{[ij][jl]}\ue^{\ui k L_{jl}} \,\,.
\end{equation}
\end{Def} 

By \eqref{eq:EF} an eigenfunction of the Laplace operator on the graph $G$ is  uniquely determined by the vector of $B=2\abs{E}$ coefficients 
$a_{[ij]}$, we will denote this vector by 
\begin{equation}
\ba\in \C^{B}\,\, .
\end{equation}
The conditions \eqref{eq:a-cond} can then be reformulated in terms of the unitary matrix $\cU_{\hat G}(k)$ acting on the vector $\ba$ as 
\begin{equation}\label{eq:ev-eq}
\cU_{\hat G}(k)\ba=\ba\,\, . 
\end{equation}
This gives a condition for the eigenvalues:  $-k^2$, $k\neq 0$,  is an eigenvalue of the Laplace operator if and only if $\cU_{\hat G}(k)$ has an eigenvalue $1$, so the 
eigenvalues are given in terms of  the roots of the secular equation
\begin{equation}
F_{\hat G}(k):=\det(\cU_{\hat G}(k)-I)=0\,\, .
\end{equation}
We will sometimes abuse notation and  refer to $k$ as well as an eigenvalue of the quantum graph. 
The eigenfunctions are then determined by the corresponding eigenvector \eqref{eq:ev-eq} of $\cU_{\hat G}(k)$. 

In Definition \ref{def:quantum_graph} we  allowed   arbitrary local S-matrices $\sigma^{(i)}$ which do not need to be associated with a 
self-adjoint extension of the Laplace operator. Then   the eigenvectors will not 
correspond eigenfunctions of a self-adjoint operator,   but one can think of the $S$-matrices as representing some 
internal dynamics in the vertices.

The vector $\ba$ determines the distribution of the state   \eqref{eq:EF} over the graph, and as a measure for how equidistributed the 
state is we will use the entropy.

\begin{Def} Let $\ba\in \C^B$, $\ba\neq 0$, then the entropy of $\ba=(a_1,a_2, \cdots, a_B)$ is defined as 
\begin{equation}
S(\ba):=\sum_{b=1}^B -\frac{\abs{a_b}^2}{\norm{\ba}^2}\ln\bigg(\frac{\abs{a_b}^2}{\norm{\ba}^2}\bigg)
\end{equation}
and the normalised entropy is 
\begin{equation}
S_N(\ba):=\frac{1}{\ln B}S(\ba)
\end{equation}
\end{Def}

For a normalised vector, $\norm \ba=1$,  the entropy is 
\begin{equation}
S(\ba):=\sum_{i=b}^B -\abs{a_b}^2\ln\abs{a_b}^2\,\, .
\end{equation}
The entropy is a measure for the distribution of the components, it satisfies
\begin{equation}
0\leq S(\ba)\leq \ln B
\end{equation}
and the two extreme cases correspond to  localisation and equidistribution. We have 
$S(\ba)=0$ if and only if all components $a_b=0$ except for one. And we have 
$S(\ba)=\ln B$ if and only if all components are equal. So the entropy is a measure for localisation or delocalisation
of the state $\ba$, in particular if $K$ elements of $\ba$ are $0$, then the entropy cannot be larger then $\ln (B-K)$, 
\begin{equation}
S(\ba)\leq \ln (B-K)\,\, .
\end{equation}
This means if the entropy is large, then $\bA$ can not be concentrated on  a small subset
Using the normalised entropy allows us to 
compare the entropy on graphs of different size.

The main tool we will use is the entropic uncertainty relation by Maassen and Uffink, \cite{MaaUff88}, which was used as well 
in \cite{AnaNon07}.

\begin{thm}[\cite{MaaUff88}] Let $U$ be a unitary $B\times B$ matrix with matrix elements $u_{bb'}$ then for any $\ba\in \C^B$ 
\begin{equation}
 S(\ba)+S(U\ba)\geq -\ln\big(\max_{b,b'}\abs{u_{b,b'}}^2\big)\,\, .
 \end{equation}
\end{thm}

 If $\ba$ happens to be an eigenvector of $U$, i.e., $U\ba=\ue^{\ui\varphi}\ba$, then $S(U\ba)=S(\ba)$, and the entropic uncertainty 
 relation gives 
 \begin{equation}
 S(\ba)\geq -\frac{1}{2} \ln \big(\max_{b,b'}\abs{u_{b,b'}}^2\big)\,\, .
 \end{equation}
Since $\ba$ is as well an eigenvector of $U^t$ for any $t\in\Z$, we obtain the 

\begin{cor}\label{cor:entr_unc_ev} Let $U$ be a unitary $B\times B$ matrix and denote the matrix elements of $U^t$, $t\in \N$, by $u^{(t)}_{b,b'}$ then for any eigenvector $\ba\in \C^B$ of 
$U$ we have 
\begin{equation}\label{eq:entr_unc_ev}
 S(\ba)\geq -\frac{1}{2} \ln\big(\max_{b,b'}\abs{u^{(t)}_{b,b'}}^2\big)\,\, .
 \end{equation}
\end{cor}

 Notice that since $U$ is unitary we have $\sum_b \abs{u_{b,b'}}^2=1$, therefore the matrix elements can not all  become arbitrary small. The smallest they can become 
 is $ \max_{b,b'}\abs{u_{b,b'}}^2=1/B$, and then all matrix elements must have the same size, and none of them can be $0$. Therefore in order to get a good estimate from 
 the entropic uncertainty relation we need a unitary matrix for which suitable powers are  not sparse. 
 
 For a quantum graph with scattering matrix $\cU_{\hat G}(k)$ this last condition can be related to the classical dynamics: Let $M_G:=(m_{b,b'})$ be defined by 
 \begin{equation}
 m_{b,b'}:=\abs{u_{b,b'}}^2=\delta_{ij} \abs{\sigma^{(i)}_{[ij],[jn]}}^2
 \end{equation}
if $b=[ij]$ and $b'=[jn]$.  Then $M$ is a doubly stochastic matrix which defines a Markov chain, and hence a random walk, on the set of oriented edges of $G$. The classical dynamics is stochastic and  is defined by 
 jumping with probability $m_{b,b'}$ from bond $b'$ to bond  $b$. Notice that these probabilities are determined  by the local $S$-matrices only.  This matrix has largest  eigenvalue  $1$ with corresponding eigenvector $\be=(1,1, \cdots, 1)^T$ and 
 so we can write 
 \begin{equation}
 M_G=\frac{1}{B} \be \be^T+R_G
 \end{equation}
 with $R_G\be=0$ and we will denote by $\mu_{\hat G}:=\norm{R_G}$ the modulus of the second largest eigenvalue. Then $\mu_{\hat G}<1$ if the graph $G$ 
 is connected  and we have 
 \begin{equation}\label{eq:class_conv}
 \biggnorm{M_G^t-\frac{1}{B} \be \be^T}\leq \mu_{\hat G}^t\,\, , 
 \end{equation}
 which means that the classical dynamics is ergodic and mixing and any probability density $M_G^t\rho$ converges exponential to the uniform distribution on the graph.

 A path, or orbit,  of length $t\in \N$ on a graph is a sequence $\gamma=(b_t,b_{t-1}, \cdots, b_1, b_0)$ of consecutive bonds, i.e., 
 if $b_s=[i,j]$ and $b_{s+1}=[k,l]$, then we must have $k=j$. We say a path $\gamma$ is without  backtracking, if $b_{s+1}\neq \hat b_s$ for all 
 $b_s\in \gamma$, and we will denote the set of all paths which go from $b'$ to $b$   in $t$ steps by $\Gamma_t(b,b')$ and the subset of paths without backtracking by 
 $\Gamma_t'(b,b')$. Then we can write for a general quantum graph 
 \begin{equation}\label{eq:orbit_sum}
 u^{(t)}_{b,b'}=\sum_{\gamma\in \Gamma_t(b,b')}\sigma_{\gamma}\ue^{\ui k L_{\gamma}}\,\, , \quad\text{where}\quad \sigma_{\gamma}=\prod_{b_s\in \gamma^{-}}\sigma_{b_{s+1},b_s}^{(i_s)} \quad \text{and}\,\,  L_{\gamma}=\sum_{b\in \gamma} L_b\,\, ,
 \end{equation}
 with $\gamma^{-}=(b_{t-1}, \cdots, b_1)$ and $i_s=b_{s+1}\cap b_s$ is the vertex connecting $b_s$ and $b_{s+1}$. If the boundary conditions prevent backtracking, then the sum is over $ \Gamma_t'(b,b')$ instead of  $\Gamma_t(b,b')$.  In order to use the entropic uncertainty principle 
 \eqref{eq:entr_unc_ev} we have to estimate $\abs{u^{(t)}_{b,b'}}^2$ which gives a double sum over paths in $ \Gamma_t(b,b')$. The diagonal terms in the sum 
 give the classical dynamics and with \eqref{eq:class_conv} we obtain  
 \begin{equation}
 \sum_{\gamma\in \Gamma_t(b,b')}\bigabs{\sigma_{\gamma}\ue^{\ui k L_{\gamma}}}^2=\big(M_{\hat G}^t\big)_{b,b'}
 =\frac{1}{B}+O\big(\mu_{\hat G}^t\big)\,\, .
 \end{equation}
 Hence if the off-diagonal terms are small for sufficiently large $t$  then we expect $\abs{u^{(t)}_{b,b'}}^2\approx \frac{1}{B}$ and so by \eqref{eq:entr_unc_ev} 
we would get  $S(\ba)\gtrapprox \frac{1}{2}\ln B$. So we have to look for quantum graphs  for which 
\begin{equation}
\abs{u^{(t)}_{b,b'}}^2=\biggabs{\sum_{\gamma\in \Gamma_t(b,b')}\sigma_{\gamma}\ue^{\ui k L_{\gamma}}}^2\approx  \sum_{\gamma\in \Gamma_t(b,b')}\bigabs{\sigma_{\gamma}\ue^{\ui k L_{\gamma}}}^2
\end{equation}
holds for sufficiently large $t$, i.e., the off diagonal contributions are small 

This leads us to the girth of a graph. The \emph{girth} $g_{G}$ of a graph $G$ is the length of the shortest cycle on $G$, where a cycle is a closed path 
without backtracking. Assume we have two paths $\gamma=(b=b_t, b_{t-1}, \cdots, b_1, b_0=b')$, $\gamma'=(b=b_t', b_{t-1}', \cdots, b_1', b_0'=b')$ of length $t$ without backtracking which connect $b$ and $b'$ and which have no bonds in common except the start and the end,  
then we can construct a closed cycle by following first $\gamma$ and then returning along $\gamma'$, $c=(b_1',\cdots , b_{t-1}',b_{t-1}, \cdots, b_1)$, this cycle has 
length $2(t-1)$ and hence we must have $2(t-1)\geq g_G$. If the two paths $\gamma,\gamma'$ have more bonds in common, then we can construct an even shorter 
cycle in the same way, therefore we find that if 
\begin{equation}\label{eq:t_bound_g}
t<\frac{g_G}{2}+1
\end{equation}
then there is at most one path (without backtracking) of length $t$  connecting any two bonds on $G$.

The girth  will be useful if we consider  equi-transmitting boundary conditions, because 
then no path with backtracking will appear when we consider powers of $\cU_{\hat G}(k)$. We will furthermore restrict ourselves as well to $d+1$ regular graphs, i.e., every vertex has degree $d+1$, 
because for these  equi-transmitting boundary conditions give $\abs{\sigma_{b,b'}^{(i)}}^2=1/d$ if $b$ follows $b'$.

\begin{thm} Let $\hat G$ be a $d+1$-regular quantum graph with equi-transmitting boundary conditions and girth $g_G$. 
Then for any eigenvector $\ba$ of $\cU_G(k)$ we have
\begin{equation}
S_N(\ba)\geq \frac{g_G \ln d}{4\ln B}\,\, .
\end{equation}
\end{thm}

\begin{proof}
We will apply the entropic uncertainty principle with $\cU_{\hat{G}}^t(k)$, where $g_G/2\leq t<g_G/2+1$. The matrix elements $u_{b,b'}^{(t)}=\sum_{\gamma \in \Gamma_t'(b,b')} \sigma_{\gamma}\ue^{\ui k L_{\gamma}}$ are given by 
sums over all paths connecting $b'$ and $b$ in $t$ steps. But by the discussion leading to \eqref{eq:t_bound_g} there is for each pair of bonds at most one such path, and hence 
\begin{equation}
\abs{ u_{b,b'}}^2\leq \abs{\sigma_{\gamma}}^2=d^{-t}\leq d^{-g_G/2}\,\, .
\end{equation}
With this the result follows from the entropic uncertainty principle. 
\end{proof}

We will now consider sequences of graphs $G_n$ such that the number of vertices $\abs{G_n}$ grows monotonically with $n$. Sequences of graphs 
whose girth growths sufficiently fast with $n$ have a special name,
 a family of $d+1$-regular graphs $G_n$, $n\in\N$, is said to have \emph{large girth} if there exist a $C>0$ with 
\begin{equation}\label{eq:large_girth}
g_{G_n}=(C+o(1))\log_{d}(\abs{G_n})\,\, ,
\end{equation}
where $\lim_{n\to\infty} o(1)=0$.  It is known that $C\leq 2$ and there 
are explicit constructions of $d+1$-regular expander families of graphs with $C=\frac{1}{2}\frac{\ln 3}{\ln(1+\sqrt{2})}$, see \cite{DavSarVal03}.  If we use that 
for a $d+1$ regular graph we have $(d+1)\abs{G}=2\abs{E}=B$, we  obtain

\begin{cor} Let $\hat G_n$ be family a $d+1$-regular quantum graphs with large girth and equi-transmitting boundary conditions. 
Then we have for any eigenvector $\ba$ of $\cU_{\hat G_n}$ that 
\begin{equation}
S_N(\ba)\geq \frac{C+o(1)}{4} \,\, ,
\end{equation}
where $C$ is the constant from \eqref{eq:large_girth}.
\end{cor}

In order to get close to the optimal bound $1/2$ one can achieve using the entropic uncertainty relation we have to ask for a very large girth, which is 
a very strong condition.

If we want to go beyond that result we have to analyse the way different terms in the orbit sum \eqref{eq:orbit_sum}
interfere if $t$ is large, i.e., if many orbits contribute. This is in general a hard problem, and to simplify it we will choose the length of the 
edges of our metric  graphs to be randomly distributed. Then the sum becomes a sum over random variables and  we can 
use Chebyshev's inequality to estimate its size.  

In addition to large girth we will need as well that the graphs are expanding. A family of graphs is called expanding if  the constant $\mu_{\hat G_n}$ which appears in \eqref{eq:class_conv} is 
uniformly bounded, i.e., there exist a $\mu<1$ such that $\mu_{\hat G_n}\leq\mu$ for all $n\in \N$. This means that the rate at which an arbitrary initial 
probability density converges to the uniform distribution is independent of the graphs size. The expansion property is typically formulated in terms of the 
spectrum of the adjacency matrix. Assume $G$ is a $d+1$ regular graph, then the normalised adjacency matrix $A_d:=\frac{1}{d}A$ is stochastic and irreducible, so 
it has an eigenvalue $1$ and all other eigenvalues have modulus less then one. Then we denote by  $\mu_G:=\max\{\abs\lambda\, ;\, \lambda \in \rm{spec}\, (A_d)\backslash \{1\}\}$ the modulus of the second largest eigenvalue. 

\begin{Def} A family of increasing $d+1$ regular graphs $G_n$ is called an expander family if there exist a $\mu<1$ such that 
\begin{equation}
\mu_{G_n}\leq \mu
\end{equation}
for all $n\in \N$.
\end{Def}

The condition in the definition is called the existence of a spectral gap. The spectral gap $1-\mu_{G_n}$ is inversely proportional to the time it takes for a random walk to 
explore the graph. For a family of expanders this time is independent of the size of the graphs. Expanders have applications in many areas, and 
have attracted therefore a lot of research, see \cite{HooLinWig06} for a review. Random $d+1$ regular graphs are with high probability expanders, so there exist a lot of them. 
But explicit constructions of concrete examples are quite involved and we refer to \cite{HooLinWig06} for more information. Expander do not necessarily have large girth, 
but random $d+1$ regular graphs have as well few short closed orbits, a fact which was used in estimates on the distribution of eigenvectors  of the discrete 
Laplacian in \cite{BroLIn13}. But there exist explicit construction of expanding graphs with large girth, see  \cite{DavSarVal03}.

Let us state our assumptions on the distribution of the lengths.

\begin{cond}\label{cond} We say that the length $L_e$, $e\in E$, are well distributed if they are independently distributed, and if there exists a $\delta>0$ and a monotonically decreasing  function $f(k)$ with $f(0)=1$ and $\lim_{k\to  \infty}f(k)=0$, such that $\Prob(L_e< \delta)=0$ and 
\begin{equation}
\abs{\E\big(\ue^{\ui k L_e}\big)} \leq f(\abs{k})\,\, .
\end{equation}
 If we have a family of graphs $G_n$, then we will require that 
this estimate holds for all $n\in\N$ with $\delta$ and $f(k)$ independent of $n$. 
\end{cond}

Notice that this condition implies that for any $\veps>0$ there exists a $k_{\veps}$ such that for all $k\geq k_{\veps}$
\begin{equation}
\abs{\E\big(\ue^{\ui k L_e}\big)}\leq \veps\,\, .
\end{equation}

Now we will assume that we have a family of  $d+1$-regular graphs $G_n$ with $\lim_{n\to\infty}\abs{G_n}=\infty$, which have large girth and 
a finite spectral gap, i.e., are expanders.  We will consider these graphs with random lengths of the edges and equi-transmittig boundary conditions.

\begin{thm} \label{thm:prob_expander} Assume $G_n$ is a family of $d+1$ regular  expanders with large girth, and $\cU_n(k)$ corresponding sequence of quantum evolution maps with 
equi-transmitting local $S$-matrices and edge lengths $\bL$ chosen randomly according to Condition \ref{cond}. Then there exists a $k_0>0$ such that for any  sequence 
$\eta_n\geq 4$,  we have 
\begin{equation}\label{eq:prob_expa}
\Prob\bigg( S_N(\ba(n))\geq \frac{1}{2}\bigg(1-\frac{\ln \eta_n}{\ln B_n}\bigg)\bigg) \geq 1-\frac{16(d+1)}{\eta_n}
\end{equation}
for any sequence of eigenvectors $\ba(n)$ of $\cU_{G_n}(k)$ with $\abs{k}\geq k_0$. 
\end{thm}

The  theorem basically states that if we consider a sequence $\ba(n)$ of eigenvectors of $\cU_{\hat{G}_n}(k)$ then 
\begin{equation}
\lim_{n\to\infty} S_N(\ba(n))\geq \frac{1}{2}
\end{equation}
holds with probability one, for $k$ large enough.  Notice that these eigenvectors don't have to have eigenvalue one, so this result is more general than just a result about eigenfunctions on the graph. 

The sequence $\eta_n$ in the statement of the theorem can be chosen in different ways depending which 
term we want to make small. E.g., if we choose $\eta_n=B_n^{\delta}$ for some $\delta >0$, then \eqref{eq:prob_expa} becomes 
\begin{equation}
\Prob\bigg( S_N(\ba(n))\geq \frac{1}{2}(1-\delta)\bigg) \geq 1-\frac{16(d+1)}{B_n^{\delta}}\,\, . 
\end{equation}
so the probability converges to $1$ reasonably fast, but the lower bound for the entropy is slightly smaller than $1/2$.  On the other hand side, if we 
want the lower bound to reach $1/2$ we have to choose a sequence $\eta_n$ which increases very slowly, e.g., the choice $\eta_n=\exp\big((\ln B_n)^{1-\delta}\big)$, 
for $\delta\in (0,1)$, gives  
\begin{equation}
\Prob\bigg( S_N(\ba(n))\geq \frac{1}{2}\big(1-(\ln B_n)^{-\delta}\big)\bigg) \geq 1-16(d+1)\ue^{-(\ln B_n)^{1-\delta}}\,\, .
\end{equation}
Now the probability converges more slowly to $1$, but the lower bound on the entropy converges to $1/2$. 

The lower bound of $1/2$ is analogous to the results obtained in \cite{AnaNon07} for manifolds of constant negative curvature.

 We found that for expanding graphs we get large entropies of the eigenfunctions, we want to compare this now with a class of quantum graphs where  
 we expect a different behaviour,  namely star graphs with Neumann boundary conditions. A star graph is a graph which has one central vertex of degree $\abs{E}$ and 
 all other vertices have degree $1$, and we will first assume Neumann boundary conditions on all vertices.
 This class of quantum graphs has been extensively studied in the literature, and 
 in \cite{KeaMarWin03,BerKeaWin04} the distribution of the eigenfunctions has been 
 investigated and it has been shown that quantum ergodicity does not hold. 
 In particular  there exist sequences of eigenfunctions which for $k\to\infty$ localise on two bonds only, therefore there exist 
 eigenfunctions whose entropy can become as small as   
 \begin{equation}
  \frac{\ln 4}{\ln B}\,\, . 
 \end{equation}
In the last section we find numerically eigenfunctions which have even smaller entropy.

Using the methods from  \cite{KeaMarWin03} and \cite{BerWin10}  we can compute a weighted energy average of the entropies of eigenfunctions on star graphs. 
Let $L_1, \cdots , L_{\abs{E}}$ be the lengths of the edges of the graph, $\bar L:=\frac{1}{\abs{E}}\sum_{i=1}^{\abs{E}} L_i$ the average length, and set for any 
$\ba\in \C^{2\abs{E}}$ with $\norm \ba=1$ 
\begin{equation}
L(\ba):= \frac{1}{\bar L} \sum_{b=1}^{2\abs{E}} L_b \abs{a_b}^2\,\, .
\end{equation}
Then our main result is for star graphs with Neumann boundary conditions is the following:

\begin{thm} \label{thm:average_S}
Let $G$ be a star graph with Neumann boundary conditions at the central vertex. Assume the bond length $L_1, \cdots, L_{\abs{E}}$ are  
linearly independent over $\Q$ and let us define  the average entropy of eigenfunctions of the star graph by 
\begin{equation}
\la S\ra(\abs E) :=\lim_{N\to\infty}\frac{1}{N}\sum_{n=1}^N \frac{1}{ L(\ba(n))}\, S_N(\ba(n))\,\, ,
\end{equation}
where $\ba(n)$ is the set of coefficients in \eqref{eq:ev-eq} associated with the $n$'th eigenfunction of the Neumann Laplacian. 
 Then
\begin{equation}
 \la S\ra(\abs{E})=\frac{C_{\rm{Neumann}}+\ln 2}{\ln \abs{E}+\ln 2}+o\bigg(\frac{1}{\ln \abs{E}}\bigg)
 \end{equation}
with
\begin{equation}
C_{\rm{Neumann}}:=\gamma+\frac{1}{\sqrt{4\pi}}\int_{-\infty}^{\infty}\ue^{-\xi^2/4}\ln m^2(\xi)\,\, \ud \xi  \,\, ,
\end{equation}
where $\gamma$ is Euler's constant and 
\begin{equation}
m(\xi)=\ue^{-\xi^2/4}+\xi \erf(\xi/2)\,\, .
\end{equation}
\end{thm}

Remark: The integral can be evaluated numerically and we find 
\begin{equation}
\frac{1}{\sqrt{4\pi}}\int_{-\infty}^{\infty}\ue^{-\xi^2/4}\ln m^2(\xi)\,\, \ud \xi=0.692032962\dots\,\, ,
\end{equation}
and so 
\begin{equation}
\la S\ra (\abs{E})=\frac{1.2692\dots +\ln 2 }{\ln \abs{E}+\ln 2}+o(1/\ln \abs{E})\,\, .
\end{equation}

If we denote the relative spread of the lengths by   $\Delta L:=\max_{e,e'\in E}\frac{\abs{L_{e}-L_{e'}}}{\bar L}$, then 
\begin{equation}
\abs{L(\ba)-1}\leq \Delta L\, \, ,
\end{equation}
 hence if $\Delta L$ is small then $\la S(\abs{E})\ra$ is close to the average entropy of eigenfunctions.

 So star graphs have very small entropies, indicating that 
eigenfunctions are on average quite localised. 
This particular behaviour of eigenfunctions on large star graphs with Neumann boundary conditions is due to   the fact that 
backscattering  is dominant for a large graph, i.e., the  bonds are only weakly coupled. The picture changes completely if we 
take equi-transmitting boundary conditions instead. Then we obtain

\begin{thm}\label{thm:star_et} Let $G$ be a star graph with $\abs{E}=B/2$ edges and equi-transmitting boundary conditions at the central vertex. Then all the 
eigenfunctions satisfy
\begin{equation}
S_N(\ba(n))\geq \frac{1}{2} \frac{\ln (B-2)}{\ln B}\,\, .
\end{equation}
\end{thm}

So we get asymptotically the strongest bound the entropic uncertainty principle allows to prove. 


\section{Regular expanding graphs}


In this section we will prove Theorem \ref{thm:prob_expander}. The proof is based on Chebyshev's inequality, so let us state it in the form we will use it:
 If $X$ is a complex valued random variable, then for any $\xi>0$ we have 
\begin{equation}
\Prob(\abs{X-\E(X)}\geq \xi)\leq \frac{\E(\abs{X}^2)-\abs{\E(X)}^2}{\xi^2}\,\, .
\end{equation}

We want to estimate the probability that $S_N(\ba)\geq \frac{1}{2}\alpha$, for some $\alpha<1/2$,  from below. By the entropic uncertainty principle, 
Corollary \ref{cor:entr_unc_ev}, we have 
\begin{equation}
\begin{split}
\Prob\bigg(S_N(\ba)\geq \frac{1}{2}\alpha\bigg)&\geq \Prob\big(\max_{b,b'} \abs{u_{bb'}^{(t)}}\leq B^{-\alpha/2}\big)\\
&=1- \Prob\big(\max_{b,b'} \abs{u_{bb'}^{(t)}}\geq B^{-\alpha/2}\big)\,\, ,
\end{split}
\end{equation}
and with $ \Prob\big(\max_{b,b'} \abs{u_{bb'}^{(t)}}\geq B^{-\alpha/2}\big)\leq \min_{b,b'} \Prob\big( \abs{u_{bb'}^{(t)}}\geq B^{-\alpha/2}\big)$ we obtain
\begin{equation}
\Prob\bigg(S_N(\ba)\geq \frac{1}{2}\alpha\bigg)\geq 1-\min_{b,b'} \Prob\big( \abs{u_{bb'}^{(t)}}\geq B^{-\alpha/2}\big)\,\, .
\end{equation}
To connect this with \eqref{eq:prob_expa} we choose $\alpha=1-\frac{\ln \eta}{\ln B}$ which gives 
\begin{equation}\label{prob:S_est}
\Prob\bigg(S_N(\ba)\geq \frac{1}{2}\bigg(1-\frac{\ln \eta}{\ln B}\bigg)\bigg)\geq 1-\min_{b,b'} \Prob\bigg( \abs{u_{bb'}^{(t)}}\geq \frac{\sqrt{\eta}}{\sqrt{ B}}\bigg)\,\, .
\end{equation}

We want to apply Chebyshev's inequality to the random variable $X=u_{b,b'}^{(t)}$ in order to estimate  $\Prob\big( \abs{u_{bb'}^{(t)}}\geq \eta^{1/2}B^{-1/2}\big)$. To that end we 
use the triangle inequality $\abs{X-\E(X)}\geq \abs{X}-\abs{\E(X)}$ to obtain 
\begin{equation}
\Prob\big( \abs{X}\geq \xi+\abs{\E(X)}\big)\leq \Prob(\abs{X-\E(X)}\geq \xi)
\end{equation}
and combining this with Chebyshevs's inequality we have 
\begin{equation}\label{eq:modCheb}
\Prob\big( \abs{X}\geq \xi+\abs{\E(X)}\big)\leq  \frac{\E(\abs{X}^2)-\abs{\E(X)}^2}{\xi^2}\leq \frac{\E(\abs{X}^2)}{\xi^2}\,\, .
\end{equation}

In order to apply this with $X=u_{b,b'}^{(t)}$  we  have to estimate the 
 expectation value and the variance. 

\begin{lem} Assume the distribution of lengths satisfies Condition \ref{cond}  and $t\geq g_G$, then 
\begin{equation}
\abs{\E(u^{(t)}_{b,b'})} \leq \frac{N_t(b,b')\big[ f(k)\big]^{g_G}}{d^{\frac{t}{2}}}
\end{equation}
and
\begin{equation}
\E(\abs{u^{(t)}_{b,b'}}^2) \leq \frac{N_t(b,b')}{d^t}\big(1+N_t(b,b')\big[ f(k)\big]^{g_G}\big)
\end{equation}
hold, where  $N_t(b,b')=\abs{\Gamma_t'(b,b')}$ denotes the number of paths connecting $b$ and $b'$ in $t$ steps without backtracking. 
\end{lem}

\begin{proof}
We have by \eqref{eq:orbit_sum}
\begin{equation}
\E(u^{(t)}_{b,b'})=\sum_{\gamma\in \Gamma_t'(b,b')} \sigma_{\gamma} \E(\ue^{\ui k L_{\gamma}})\,\, .
\end{equation}
Now we observe that if  $\gamma\in\Gamma_t'(b,b')$ visits a bond $b''$ twice, then $\gamma$ must contain a cycle, because if $\gamma$ does not contain a cycle 
then it can only visit $b''$ twice by going backwards, but backtracking is prohibited in $\Gamma_t'(b,b')$.  So since $t\geq g_{G}$ there are at least $g_{G}$ different bonds in 
$\gamma$, because $g_G$ is the number of bonds in the shortest cycle. 
 So if we write $L_{\gamma}=\sum_{e\in E} g_e(\gamma)L_e$, where $g_e(\gamma)\in \N_0$ denotes the 
number of times $e$ is visited by the path $\gamma$, then 
\begin{equation}
\E(\ue^{\ui k L_{\gamma}})=\prod_{e\in E} \E(\ue^{\ui g_e(\gamma) k L_e})
\end{equation}
Since $\E(1)=1$ and $g_e(\gamma)\geq 1$ for at least $g_G$ different edges, we obtain from Condition \ref{cond}
\begin{equation}
\abs{\E(\ue^{\ui k L_{\gamma}})}\leq f(k)^{g_G}\,\, ,
\end{equation}
and hence 
\begin{equation}
\abs{\E(u^{(t)}_{b,b'})}\leq  \frac{N_t(b,b') f(k)^{g_G}}{{d}^{t/2}}\,\, ,
\end{equation}
where we have used as well that $\abs{\sigma_{\gamma}}={d}^{-t/2}$.

The variance we estimate using the same ideas: we first split the double sum into a diagonal and off-diagonal part
\begin{equation}
\E(\abs{u^{(t)}_{b,b'}}^2)=\sum_{\gamma\in\Gamma_t'(b,b')} \abs{\sigma_{\gamma}}^2+\sum_{\gamma\neq\gamma'\in \Gamma_t'(b,b')} \sigma_{\gamma}\sigma_{\gamma'}^*\E\big(\ue^{\ui k (L_{\gamma}-L_{\gamma'})}\big)\,\, , 
\end{equation}
and the diagonal part is just $\sum_{\gamma\in \Gamma_t'(b,b')} \abs{\sigma_{\gamma}}^2 =d^{-t} N_t(b,b')$.  For the off-diagonal terms we use that 
$\gamma$ and $\gamma'$ must differ on  at least $g_G$  edges, otherwise $\gamma\cup \gamma'$ would contain a closed cycle of length less then $g_G$. 
 Then $L_{\gamma}-L_{\gamma'}=\sum_{e\in \hat E} (g_{\gamma}(e)-g_{\gamma'}(e))L_e$ 
and $\abs{g_{\gamma}(e)-g_{\gamma'}(e)}\geq 1$ for at least $g_G$ edges, hence 
\begin{equation}
\biggabs{\sum_{\gamma\neq\gamma'} \sigma_{\gamma}\sigma_{\gamma'}^*\E\big(\ue^{\ui k (L_{\gamma}-L_{\gamma'})}\big)}\leq N_t(b,b')^2 d^{-t} f(t)^{g_G}\,\, . 
\end{equation}
So combining the estimates for the two terms gives 
\begin{equation}
\E(\abs{u^{(t)}_{b,b'}}^2)\leq \frac{N_t(b,b')}{d^t}\big(1+N_t(b,b') f(k)^{g_G}\big)\,\, .
\end{equation}
\end{proof}

Let us now consider the number of paths connecting $b$ and $b'$,  $N_t(b,b')$.

\begin{lem} Let $A_d=\frac{1}{d}A $ be the normalised adjacency matrix  of a $d+1$ regular graph $G$ and let $\mu_G:=\max\{\abs{\lambda}\, \lambda\in \sigma(A_d)\backslash\{1\}\} $ be the spectral gap to the leading eigenvalue $1$, then we have 
\begin{equation}
N_t(b,b')\leq \frac{d^t}{\abs{G}}\big(1+\abs{G}\mu_G^t\big)\,\, .
\end{equation}
\end{lem}

\begin{proof}
Let $b=[i,j]$ and $b'=[i',j']$ and let $n_t(i,i')$ be the number of paths connecting the vertices $i,i'$ in $t$ steps. Then 
\begin{equation}
N_t(b,b')\leq   n_t(i,i')\,\, , 
\end{equation}
and so to obtain an upper bound on $N_t(b,b')$ it is enough to estimate $n_t(i,i')$. 
To this end we use that
\begin{equation}
 n_t(i,i')=[A^t]_{i,i'}=\be_i\cdot A^t\be_{i'}
 \end{equation}
 where $A$is the adjacency matrix of $G$ and $\be_i\in \C^{\abs{G}}$, $i=1, \cdots , \abs{V}$,  denote the canonical basis vectors. 
 Now $A$ is a symmetric matrix with leading eigenvalue $d$ and corresponding normalised eigenvector $\frac{1}{\sqrt{\abs G}}\be $ where 
 $\be =(1,1, \cdots , 1)^T$, and by the spectral theorem 
 \begin{equation}
 A^t=\frac{d^t}{\abs{G}} \be \be^T+ A_R^t
 \end{equation}
 with $\norm{A_R}=d\mu_G$ and $ \norm{A_R^t}=(d\mu_G)^t$. If we apply this to the expression for  $n_t(i,i')$ we obtain 
 \begin{equation}
 n_t(i,i')=\be_i\cdot A^t\be_j=\frac{d^t}{\abs{G}}+\be_i\cdot A_1^t\be_j\leq \frac{d^t}{\abs{G}}+(d\mu_G)^t
 \end{equation}
 \end{proof}

The assumption that we have a finite 
spectral gap means that there exist a $\mu<1$, independent of $n$, such that for all graphs in the sequence $G_n$ we have 
$\mu_G\leq \mu$. Now we choose $t$ to be the smallest integer such that 
\begin{equation}
\mu^t\leq \frac{1}{\abs{G}}
\end{equation}
so that 
\begin{equation}
N_t(b,b')\leq 2\frac{d^t}{\abs{G}}\,\, ,
\end{equation}
note that this means
\begin{equation}
t\geq \frac{1}{\ln \mu_G}\ln \abs{G_n}\,\, .
\end{equation}
With this choice of $t$ we have 
\begin{equation}
\E(\abs{u^{(t)}_{b,b'}}^2)\leq \frac{2}{\abs{G}}\big(1+N_t(b,b') f(k)^{g_G}\big)\,\, ,\quad \text{and}\quad\abs{\E(u^{(t)}_{b,b'})}\leq  \frac{d^{t/2} f(k)^{g_{G}}}{\abs{G}}\,\, .
\end{equation}

We have fixed $t$ now, so the only choice left is the size of $k$. To this end we will use that for a $d+1$ regular graph $B=2\abs{E}=(d+1)\abs{G}$. Since we have large girth, i.e., $g_{G}=C\ln\abs{G}$ for some $C>0$, there exist a $\veps>0$ such that 
\begin{equation}
\frac{d^{t/2} \veps^{g_{G_n}}}{\abs{G_n}} \leq \frac{1}{B^{1/2}}\,\, \quad\text{and}\quad N_t \veps^{g_G}\leq 1\,\, .
\end{equation}
Hence we choose $k_0$ such that $f(k_0)=\veps$, and  we obtain for $k\geq k_0$ 
\begin{equation}\label{eq:mean_var_est}
\E(\abs{u^{(t)}_{b,b'}}^2)\leq \frac{4(d+1)}{B} ,\, \quad\text{and}\quad \abs{\E(u^{(t)}_{b,b'})}^2\leq \frac{1}{B}\,\, . 
\end{equation}

Inserting these estimate into Chebyshev's inequality \eqref{eq:modCheb} we find 
\begin{equation}
\Prob\big(\abs{u^{(t)}_{b,b'}}\geq \xi+B^{-1/2}\big)\leq \Prob\big(\abs{u^{(t)}_{b,b'}}\geq \xi+\abs{\E(u^{(t)}_{b,b'})}\big)\leq \frac{4(d+1)}{\xi^2B}\,\, .
\end{equation}
Now in view of \eqref{prob:S_est} we choose $\xi$ such that $\xi+B^{-1/2}=\eta^{1/2}B^{-1/2}$, i.e, 
\begin{equation}
\xi=B^{-1/2}(\eta^{1/2}-1)
\end{equation}
and if $\eta\geq 4$ we have $\xi\geq 1 $ and $(\xi^2 B)^{-1}\leq 4\eta^{-1}$, hence we found 
\begin{equation}
\Prob\big(\abs{u^{(t)}_{b,b'}}\geq \eta^{1/2}B^{-1/2}\big)\leq \frac{16(d+1)}{\eta}\,\, .
\end{equation}
But combining this estimate with \eqref{prob:S_est} gives Theorem \ref{thm:prob_expander}.


\section{Star Graphs}

The statistical properties of eigenvalues and eigenfunctions on star graphs with Neuman boundary conditions have been studied quite in some detail. 
We will use the results from \cite{KeaMarWin03} to compute the average entropy of eigenfunctions.

Let us first recall that on a general star graph with Neuman boundary conditions on the end of the edges, but arbitrary 
boundary conditions on the central vertex, we can always write the $n$'th eigenfunction on edge $e$ as 
\begin{equation}
\psi_e^{(n)} (x)=A_e(n)\cos(k_n(x-L_e)) \,\, ,
\end{equation}
with $L_e$ be the length of edge $e$,  $k_n^2$ the $n$'th eigenvalue. Hence on a star graph with $\abs{E}$ edges, the eigenfunctions are determined by 
a vector $\bA(n)\in \C^{\abs{E}}$ of half the size compared to a general graph. The normalisation is chosen such that 
\begin{equation}
\norm{\bA(n)}^2=\sum_{e=1}^{\abs{E}} \abs{A_e(n)}^2=1
\end{equation}
holds, and so we can define another  entropy of the $n$'th eigenstate as
\begin{equation}
S_N(\bA):=\frac{1}{\ln \abs{E}}\sum_{e=1}^{\abs{E}} -\abs{A_e}^2\ln \abs{A_e}^2\,\, .
\end{equation}

Given $\bA$ we can easily find the vector $\ba\in \C^{2\abs{E}}$ which we use in the general case  for a characterisation, since 
\begin{equation}
\begin{split}
\psi_e(x)=A_e\cos(k_n(x-L_e))&=\frac{A_e\ue^{-\ui k_n L_e}}{2}\ue^{\ui k_n x}+\frac{A_e\ue^{\ui k_n L_e}}{2}\ue^{-\ui k_n x}\\
&=a_e^{(in)}\ue^{\ui k_n x}+a_e^{(out)}\ue^{-\ui k_n x} 
\end{split}
\end{equation}
with 
\begin{equation}\label{eq:a_A}
a_e^{(in)}=\frac{A_e\ue^{-\ui k_n L_e}}{2}\,\, ,\quad a_e^{(out)}=\frac{A_e\ue^{\ui k_n L_e}}{2} \,\, .
\end{equation}
We can use this to compare the different entropies:

\begin{lem} \label{lem:Aa}
We have 
\begin{equation}\label{eq:Sa_SA}
S_N(\ba)=\frac{\ln \abs{E}}{\ln \abs{E}+\ln 2} S_N(\bA)+\frac{\ln 2}{\ln \abs{E}+\ln 2} \,\, .
\end{equation}
\end{lem}

\begin{proof}
If $\ba\in \C^{2\abs{E}}$ is not normalised we have 
\begin{equation}
S_N(\ba)=\frac{1}{\ln(2\abs{E})}\sum_{b=1}^{2\abs{E}}-\frac{\abs{a_b}^2}{\norm{\ba}^2}\ln \frac{\abs{a_b}^2}{\norm{\ba}^2}
\end{equation}
and  with the relations \eqref{eq:a_A} we then find $\norm{\ba}^2=\norm{\bA}^2/2=1/2$, hence $\frac{\abs{a_{b_i}}^2}{\norm{\ba}^2}
=\abs{A_i}^2/2$ and so 
\begin{equation}
S_N(\ba)=\frac{1}{\ln(2\abs{E})}2\sum_{i=1}^{\abs{E}}-\frac{\abs{A_i}^2}{2}\ln \frac{\abs{A_i}^2}{2}=\frac{\ln \abs{E}}{\ln(2\abs{E})} S_N(\bA)+\frac{\ln 2}{\ln 2\abs{E}}\,\, .
\end{equation}
 \end{proof}

Notice that the relation \eqref{eq:Sa_SA} is a convex combination interpolating between $S(\bA)$ and $1$, therefore $S_N(\ba)$ is always larger 
than $S_N(\bA)$ which means that if we can find lower bounds for both  $S_N(\ba)$ and $S_N(\bA)$ of similar size, then a lower bound on $S_N(\bA)$ will give a stronger estimate.

We can use \eqref{eq:a_A} as well to write down an eigenvector equation for $\bA$. Let $\sigma_0$ be the the S-matrix related to the boundary conditions at the central vertex, then 
\begin{equation}
\sigma_0\ba^{(in)}=\ba^{(out)}\,\, . 
\end{equation}
and together with  \eqref{eq:a_A} this  gives 
\begin{equation}
\ue^{\ui k_n\bL} \sigma_0\ue^{\ui k_n\bL} \bA(n)=\bA(n)\,\, ,
\end{equation}
where $\ue^{\ui k_n\bL}$ denotes the diagonal matrix with diagonal elements $\ue^{\ui k_n L_e}$, $e=1, \cdots , \abs{E}$. Notice that the matrix 
\begin{equation}\label{eq:sigma_k}
\sigma_k:=\ue^{\ui k_n\bL} \sigma_0\ue^{\ui k_n\bL}
\end{equation}
 is unitary, and hence we can apply the entropic uncertainty principle to obtain
\begin{equation}\label{eq:star_SS}
S(\bA(n))\geq -\frac{1}{2\ln \abs{E}}\, \ln \big( \max_{ee'}\abs{{\sigma_k}_{e,e'}}^2\big)\,\, .
\end{equation}

\subsection{Neumann boundary conditions}

For a star graph with Neumann boundary conditions at the central vertex we have for $\abs{E}\geq 4$
\begin{equation}
 \max_{ee'}\abs{{\sigma_k}_{e,e'}}^2=\bigg(1-\frac{2}{\abs{E}}\bigg)^2
 \end{equation}
 hence the entropic uncertainty principle gives 
 \begin{equation}
\ln \abs{E}\,  S_N(\bA(n))\geq -\ln \bigg(1-\frac{2}{\abs{E}}\bigg)= \frac{2}{\abs{E}}+O\bigg(\frac{1}{\abs{E}^2}\bigg)\,\, .
\end{equation}
The Neumann boundary conditions imply that an eigenfunction cannot be concentrated on a single edge, one needs at least two edges for the support, and 
in case the length are linearly dependent over $\Z$ one can construct explicitly examples of eigenfunction which are concentrated on two edges only, with equal weight on both edges. For these the entropy is 
\begin{equation}
S_N(\bA(n))=\frac{\ln 2}{\ln \abs{E}} \,\, ,
\end{equation}
and we expect that this is the smallest value the  entropy of an eigenfunction on the Neumann star graph can take. In \cite{KeaMarWin03} it was shown that 
even on graphs where the length of the bonds are rationally independent one can construct eigenfunction who for large $k$ concentrate on two edges. 
So the entropic uncertainty principle doesn't give a good bound for Neumann star graphs. One could try to improve on this by using powers of $\sigma_k$, but we will follow a different route instead and compute the average entropy.

For a star graph with Neumann boundary conditions at the central vertex the coefficients $A_i(n)$ can be expressed directly in terms of $k_n$ as 
\begin{equation} 
\abs{A_i(n)}^2=\frac{\sec^2(k_n L_i)}{\sum_{i=1}^{\abs{E}} \sec^2(k_nL_i)} \,\, 
\end{equation}
where the sum in the denominator ensures that the vector $\bA(n)$ is normalised. This explicit form has been used in \cite{KeaMarWin03} to study the 
distribution of the $A_i(n)$, and we will use the methods from that paper to compute the average entropies for large energies and large degree $\abs{E}$. 
In \cite{KeaMarWin03} it is assumed that the length $L_1, \cdots ,L_{\abs{E}}$ are linearly independent and lie in an small interval 
$[\bar L-\Delta L/2,\bar L+\Delta L/2]$ with $\abs{E}\Delta L\to 0$ for $\abs{E}\to\infty$, we like to relax this condition by using the results from \cite{BerWin10}.

\begin{lem} \label{lem:indepL}
Assume the length  $L_1, \cdots , L_{\abs{E}}$ are linearly independent over $\Q$ and set 
\begin{equation}
L(\bA):=\frac{1}{\bar L} \sum_{i=1}^{\abs{E}} L_i \abs{A_i(n)}^2 
\end{equation}
where $\bar L=\frac{1}{\abs{E}}\sum_{i=1}^{\abs{E}}L_i$. Then the limit 
\begin{equation}
\lim_{N\to\infty}\frac{1}{N} \sum_{n=1}^N \frac{S_N(\bA(n))}{L(\bA(n))} 
\end{equation}
is independent of the length $\bL=(L_1, \cdots , L_{\abs{E}})$. 
\end{lem}

The proof of this lemma follows using the methods in \cite{BerWin10}: The function $S(\bx)/L(\bx)$ satisfies the conditions\footnote{Notice that in Lemma 5.1. of \cite{BerWin10}, the conditions on $G(\bx)$ should  include that it is gauge invariant, i.e., $G(\ue^{\ui\alpha}\bx)=G(\bx)$ for all $\alpha\in [0,2\pi)$ and $\bx\in \C^{\abs{E}}$. Otherwise the 
function $\Phi(\bx)$ is not well defined. But $S(\bx)/L(\bx)$ satisfies this relation.}
 required of $G$ in Lemma 5.1 of 
\cite{BerWin10}, and as a consequence the proof of Theorem 3.4 can be extended from covering a weighted average over  moments to 
 covering a weighted average over entropies.

If the length $L_i$ have a small spread $\Delta L$, i.e., if $L_i/\bar L\in [1-\Delta L /2, 1+\Delta L/2]$, then 
\begin{equation}
 \frac{1}{N}\sum_{n=1}^N \frac{S_N(\bA(n))}{L(\bA(n))} =\frac{1}{N} \sum_{n=1}^N S_N(\bA(n)) +O(\Delta L)\,\, 
 \end{equation}
 hence for small $\Delta L$ the weighted average is close to the average entropies. Now we use the results from  \cite{KeaMarWin03} 
 to compute the average.

\begin{thm} \label{thm:star_average}
Let the bond length $L_1, \cdots, L_{\abs{E}}$ be linearly independent over $\Q$, then the entropies of eigenfunctions of the star graph with Neumann boundary condition   satisfy
\begin{equation}
\lim_{N\to\infty}\frac{1}{N}\sum_{n=1}^N \frac{S_N(\bA(n))}{L(\bA(n))}=\la S_N(\abs{E})\ra  
\end{equation}
with 
\begin{equation}
\lim_{\abs{E}\to\infty}\ln(\abs{E}) \, \la S_N(\abs{E})\ra=\gamma+\frac{1}{\sqrt{4\pi}}\int_{-\infty}^{\infty}\ue^{-\xi^2/4}\ln m^2(\xi)\,\, \ud \xi\,\, . 
\end{equation}
where $\gamma$ is Euler's constant and 
\begin{equation}
m(\xi)=\ue^{-\xi^2/4}+\xi \erf(\xi/2)\,\, .
\end{equation}
\end{thm}

\begin{proof}
We will use $v=\abs{E}$ during the proof to save space. 
Let us recall  two results from \cite{KeaMarWin03} on which the proof is based. First, if $f:\R^v/(\pi\Z)^v\to\R$ is a piecewise continuous function then 
\begin{equation}\label{eq:gen_average}
\begin{split}
\lim_{N\to\infty}&\frac{1}{N}\sum_{n=1}^N f(k_n\bL)\\ &=\frac{1}{2\pi^v v\bar L} \int_{-\infty}^{\infty} \int_0^{\pi}\cdots \int_0^{\pi} f(\bx) \sum_{i=1}^vL_i\sec^2(x_i) \ue^{\ui\zeta \sum_{i=1}^v \tan x_i} \ud x_1\cdots \ud x_v\ud \zeta \,\, ,
\end{split}
\end{equation}
this follows by combining Theorem 8 and equation $(15)$ in  \cite{KeaMarWin03}.

The second result we use is on the distribution of $\sum_{i=1}^v \sec^2(k_nL_i)$: There exist a probability  density $P_v(y)$ with $P_v(y)=0$ for $y\leq 0$ such that 
 for any 
continuous function with compact support $\varphi$,  we have 
\begin{equation}\label{eq:spec_average}
\lim_{N\to\infty} \frac{1}{N}\sum_{n=1}^N\varphi\bigg(\frac{1}{v^{2}}\sum_{i=1}^v \sec^2(k_nL_i)\bigg)=\int P_v(y) \varphi(y)\, \ud y \,\, .
\end{equation}
Furthermore for $y\geq 0$ 
\begin{equation}
P(v):=\lim_{v\to\infty} P_v(y)= \frac{1}{4\pi y^{3/2}} \int_{-\infty}^{\infty} \ue^{-\xi^2/4- m(\xi)^2/(4y)} m(\xi)\, \ud \xi\,\, ,
\end{equation}
where
\begin{equation}
m(\xi)=\frac{2}{\sqrt{\pi}} \ue^{-\xi^2/4}+\xi \erf(\xi/2)\,\, .
\end{equation}
 These are Theorems 3 and 4 in  \cite{KeaMarWin03}, notice that we don't need that $v\Delta L\to 0$, this condition comes in \cite{KeaMarWin03} from the fact that they consider 
 the distribution of $\sum L_j \sec^2(k_nL_j)$ and approximate 
 $\sum L_j \sec^2(k_nL_j)$ by $\sum \sec^2(k_nL_j)$, but we will be  interested in $\sum \sec^2(k_nL_j)$ only.

Now let us start by rewriting the normalised entropy in the form 
\begin{equation}
\begin{split}
\frac{S_N(\bA(n))}{L(\bA(n))}=&\frac{1}{\ln v}\frac{\bar L}{\sum_{i=1}^v L_i\sec^2(k_nL_i)}\sum_{i=1}^v -\sec^2(k_nL_i)\ln \sec^2(k_nL_i)\\
&\qquad \qquad+\frac{1}{\ln v} \frac{\bar L \sum_{i=1}^v \sec^2(k_nL_i)}{\sum_{i=1}^v L_i\sec^2(k_nL_i)} \ln\bigg(\sum_{i=1}^v \sec^2(k_nL_i)\bigg) 
\end{split}
\end{equation} 
and then using \eqref{eq:gen_average} for the first term and \eqref{eq:spec_average} for the second term we obtain 
\begin{equation}
\lim_{N\to\infty}\frac{1}{N}\sum_{n=1}^N S_N(\bA(n))=A(v)+B(v)
\end{equation}
where 
\begin{equation}
A(v)=\frac{1}{2\pi^v v\ln v} \int_{-\infty}^{\infty} \int_0^{\pi}\cdots \int_0^{\pi}  \sum_{i=1}^v-\sec^2(x_i)\ln \sec^2(x_i) \ue^{\ui\zeta \sum_{j=1}^v \tan x_j} \ud x_1\cdots \ud x_v\ud \zeta
\end{equation}
and
\begin{equation}
B(v)=\frac{1}{\ln v}\int P_v(y)\ln(v^2 y)\, \ud y=2+\frac{1}{\ln v}\int P_v(y)\ln( y)\, \ud y\,\, .
\end{equation}
Her we have used  \eqref{eq:spec_average} for a function without compact support, we should replace this by a compact approximation, but since 
we later take the limit $v\to \infty$, and $P(v)$ is flat at $0$ and decays algebraically at $\infty$, the resulting integrals are well defined.

Using the integral 
\begin{equation}
\frac{1}{\pi}\int_0^{\pi} \ue^{\ui\zeta \tan x}\, \ud x=\ue^{-\abs{\zeta}} 
\end{equation}
we can reduce the first part to 
\begin{equation}
\begin{split}
A(v)&=\frac{-1}{2\pi \ln v} \int_{-\infty}^{\infty} \int_0^{\pi}\sec^2(x)\ln \sec^2(x) \ue^{\ui\zeta \tan x}\ue^{-(v-1)\abs{\zeta}} \ud x\, \ud \zeta\\
&=\frac{-1}{\pi \ln v}\int_0^{\pi}\sec^2(x)\ln \sec^2(x) \frac{(v-1)}{(v-1)^2+\tan^2(x)} \,\, \ud x
\end{split}
\end{equation}
where we have used as well that 
\begin{equation}
\int_{-\infty}^{\infty}\exp(\ui \zeta \tan x-(v-1)\abs{\zeta})\, \ud \zeta =\frac{2(v-1)}{(v-1)^2+\tan^2 x}\,\, . 
\end{equation}
Finally we perform the substitution $z=\tan x$, and we arrive at 
\begin{equation}
A(v)= \frac{-1}{\pi \ln v}2\int_0^{\infty} \frac{(v-1)\ln(1+z^2)}{(v-1)^2+z^2}\, \ud z= \frac{-1}{\pi \ln v}2\pi \ln v=-2\,\, ,
\end{equation}
with the help of $\int_0^{\infty}  \frac{a \ln(1+z^2)}{a^2+z^2}\, \ud z=\pi \ln(a+1)$ (See 4.295 in  \cite{GraRys00}). 

Let us turn to the second term, $B(v)$, we have 
\begin{equation}
B(v)=2+\frac{1}{\ln v}\int P_v(y)\ln( y)\, \ud y
\end{equation}
and so the limit $v\to \infty$ gives 
\begin{equation}
\lim_{v\to\infty}\ln v\,  \la S_N(v) \ra =\int_0^{\infty} P(y)\ln(y)\, \ud y\,\, .
\end{equation}
If we insert the formula for $P(y)$ and exchange the order of integration, the $y$-integral becomes 
\begin{equation}
\begin{split}\int_0^{\infty}\exp\bigg(-\frac{1}{4y} m(\xi)^2\bigg)\frac{\ln(y)}{y^{3/2}}\, \ud y
&=\frac{-2}{m(\xi)}\int_0^{\infty}\ue^{-s}  \frac{\ln( sm^{-2}(\xi)4)}{s^{1/2}}\, \ud s\\
&=\frac{-2}{m(\xi)}\int_0^{\infty}\ue^{-s} s^{-1/2} \ln s\, \ud s\\
&\quad +\frac{2\ln( m^2(\xi)/4)}{m(\xi)}\int_0^{\infty}\ue^{-s}  s^{-1/2}\, \ud s\\
&=\frac{-2}{m(\xi)}\Gamma'(1/2)+\frac{2\ln( m^2(\xi)/4)}{m(\xi)}\Gamma(1/2)\\
&=\frac{\sqrt{4\pi}}{m(\xi)} [\gamma +\ln(m(\xi)^2)]
\end{split}
\end{equation}
where  $\gamma$ is Euler's constant and we used $\Gamma(1/2)=\sqrt{\pi}$ and $\Gamma'(1/2)=-(\gamma+\ln 4)\sqrt{\pi}$. Hence 
\begin{equation}
\begin{split}
\lim_{v\to\infty}\ln v \la S_N(v) \ra&= \frac{1}{\sqrt{4\pi} } \int_{-\infty}^{\infty} \ue^{-\xi^2/4}[\gamma +\ln(m^2(\xi))]\, \ud \xi\\
&=\gamma +\frac{1}{\sqrt{4\pi} } \int_{-\infty}^{\infty} \ue^{-\xi^2/4}\ln(m^2(\xi))\, \ud \xi\,\, .
\end{split}
\end{equation}
\end{proof}

Finally we want to use Lemma \ref{lem:Aa} to relate the entropy for $\bA$ to the entropy for $\ba$. We notice first that the relations between $\bA$ and $\ba$ we 
used in the proof of Lemma  \ref{lem:Aa} give us that 
\begin{equation}
L(\ba)=L(\bA)\,\, .
\end{equation}
Secondly, we use that Theorem 3.4 with $m=0$ in \cite{BerWin10} gives us that 
\begin{equation}
\lim_{N\to\infty} \frac{1}{N}\sum_{n=1}^N \frac{1}{L(\ba(n))}=1
\end{equation}
and therefore Lemma  \ref{lem:Aa}  gives
\begin{equation}
\lim_{N\to\infty}\frac{1}{N} \sum_{n=1}^N \frac{S_N(\ba(n))}{L(\ba(n))} =\frac{C_{\rm{Neumann}}+\ln 2}{\ln v+\ln 2}+o\bigg(\frac{1}{\ln v}\bigg)
\end{equation}
which is Theorem \ref{thm:average_S}.

\subsection{Equi-transmitting boundary conditions}

For equi-transmitting boundary conditions we can use \eqref{eq:star_SS} to get directly an optimal  lower bound on the entropy.

\begin{thm} Let $\hat G$ be a star graph with $\abs{E}$ edges and equi-transmitting boundary conditions at the central vertex. Then all the 
eigenfunctions satisfy
\begin{equation}
S_N(\bA(n))\geq \frac{1}{2} \frac{\ln (\abs{E}-1)}{\ln \abs{E}}\,\, .
\end{equation}
\end{thm}

\begin{proof}
Since the boundary conditions are equi-transmitting the $S$ matrix elements satisfiy $\abs{{S}_{ij}}^2=(1-\delta_{ij})(\abs{E}-1)^{-1}$. Hence \eqref{eq:star_SS}
gives 
\begin{equation}
S(\bA(n))=\frac{1}{2} \frac{\ln (\abs{E}-1)}{\ln \abs{E}}\,\, .
\end{equation}
\end{proof}

This result is optimal in the sense that we get for $\abs{E}\to\infty$ the best bound which can be obtained using the entropic uncertainty principle.

Combing this with Lemma \ref{lem:Aa} gives Theorem \ref{thm:star_et}.


\section{Comparison with Numerical Results}

\begin{figure}[t] 
\begin{center}
\includegraphics[width=8cm, height=4cm]{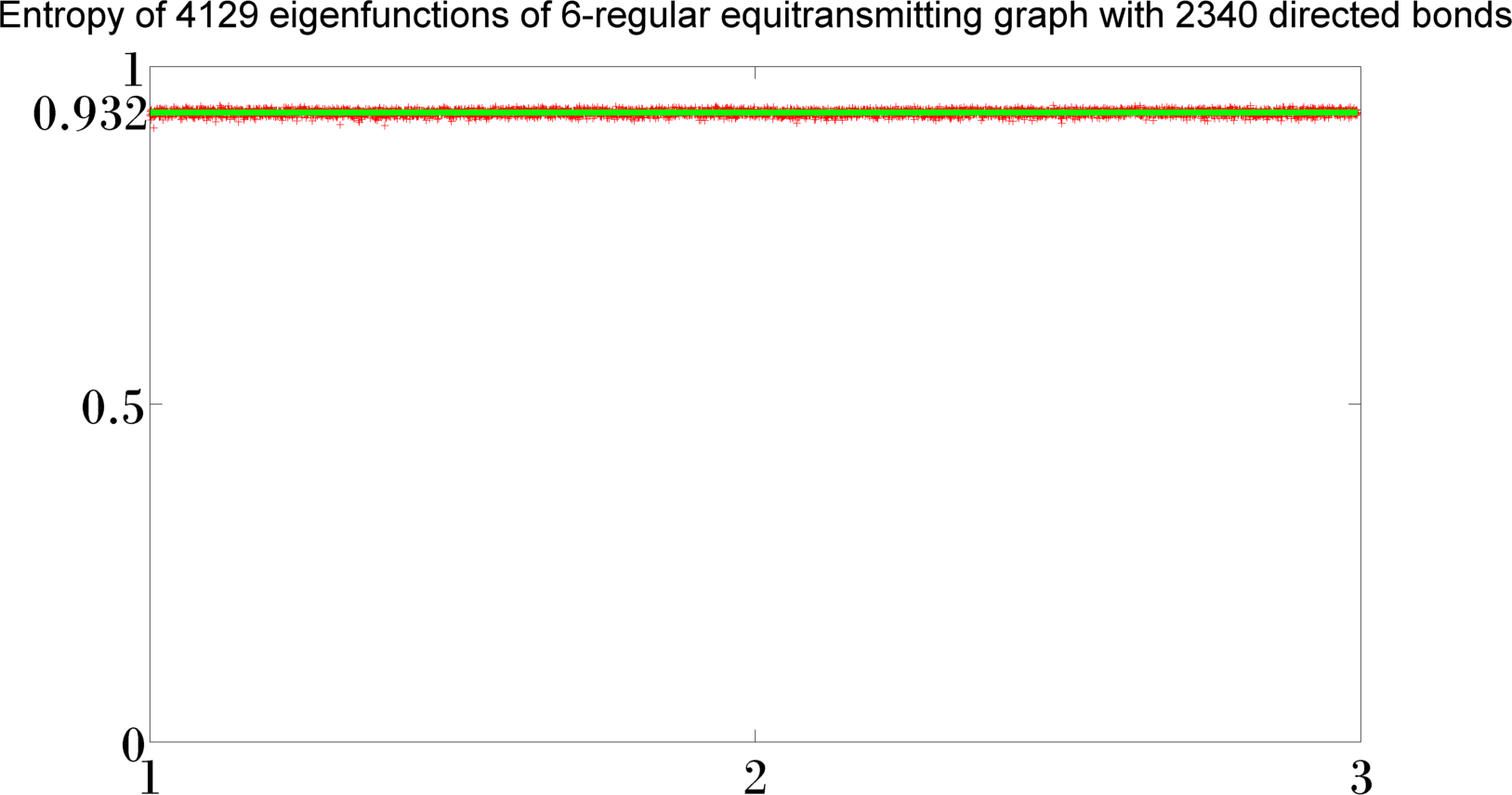}
\includegraphics[width=8cm, height=4cm]{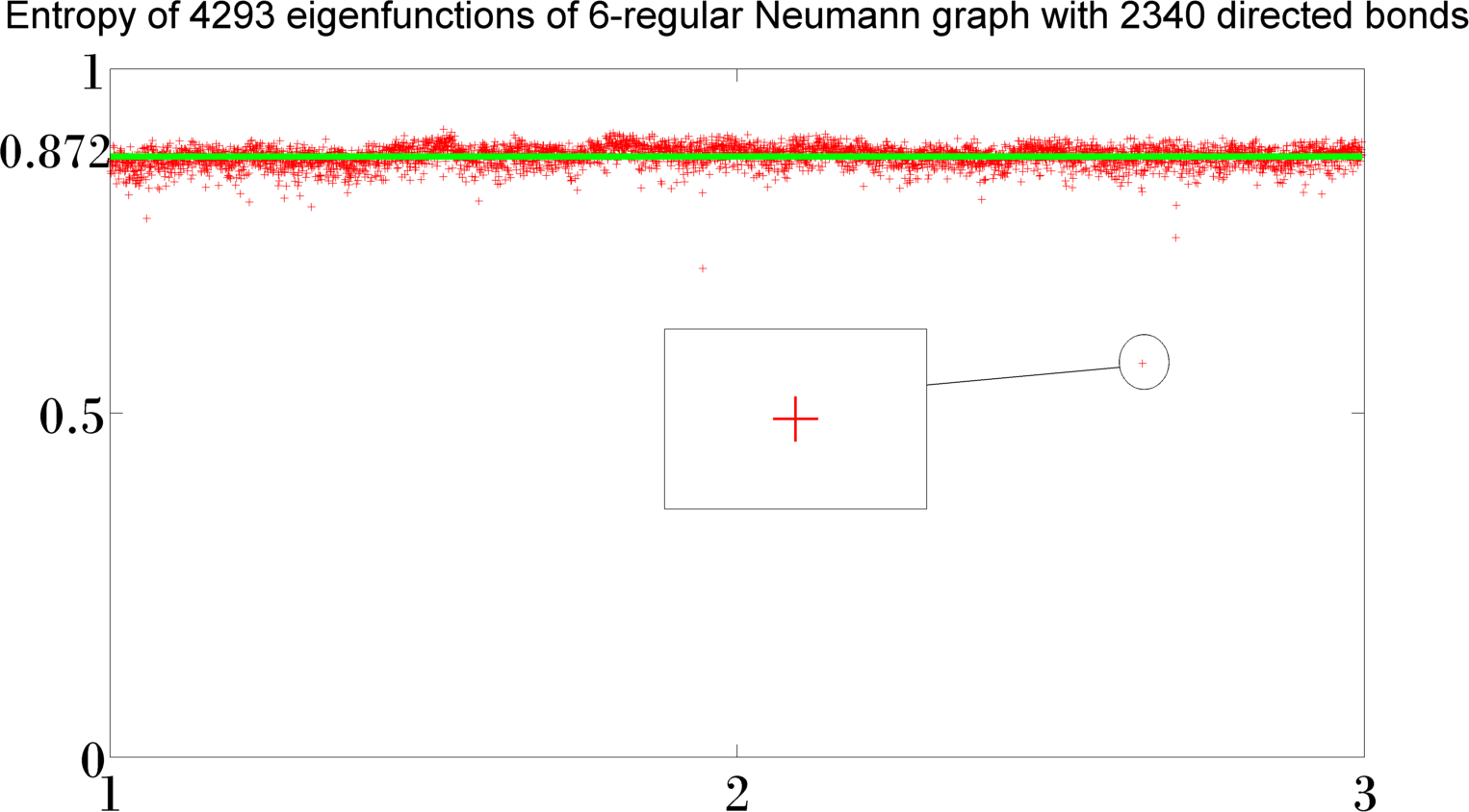}
\end{center}
\begin{center}
\includegraphics[width=8cm, height=4cm]{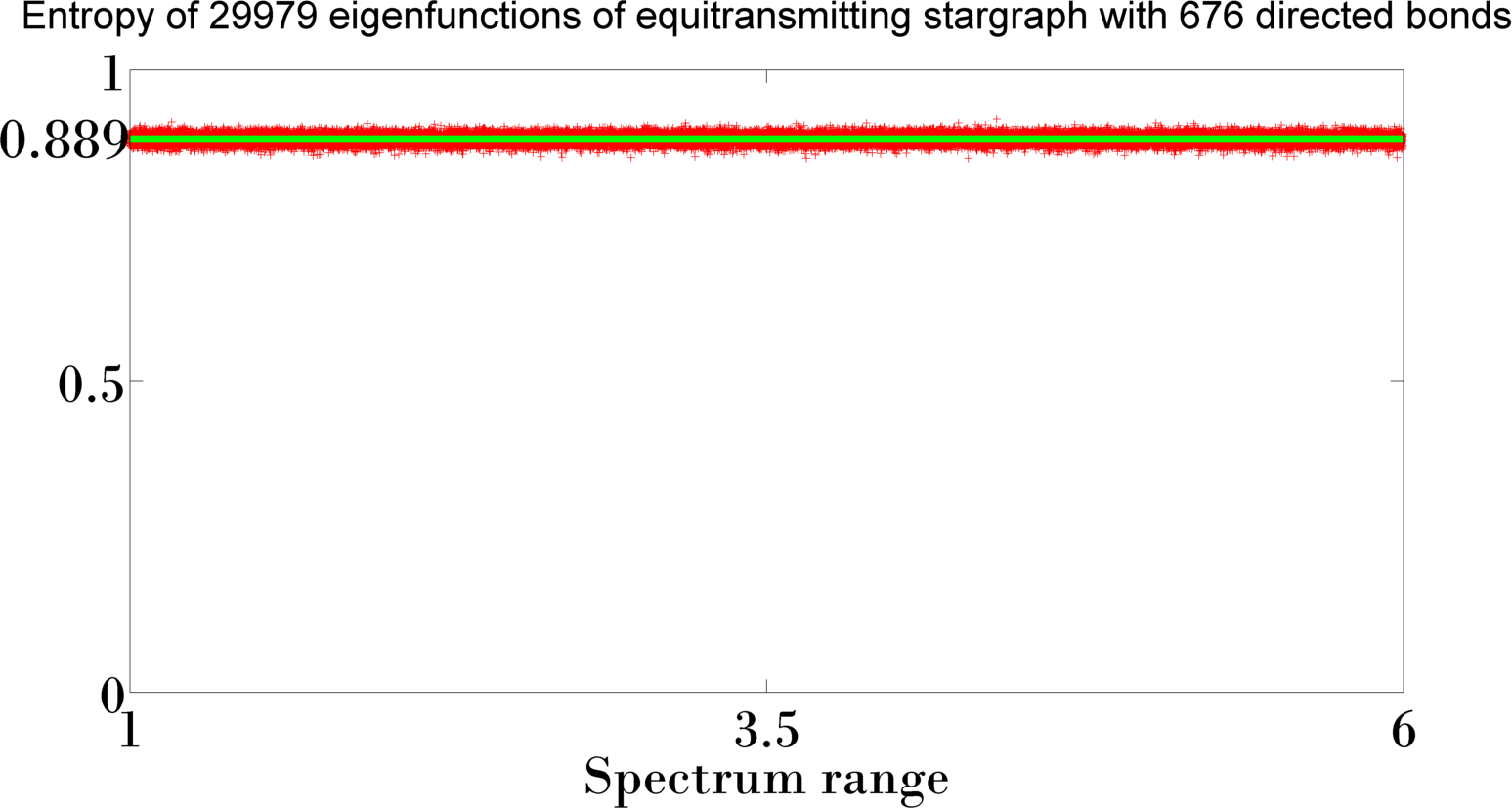}
\includegraphics[width=8cm, height=4cm]{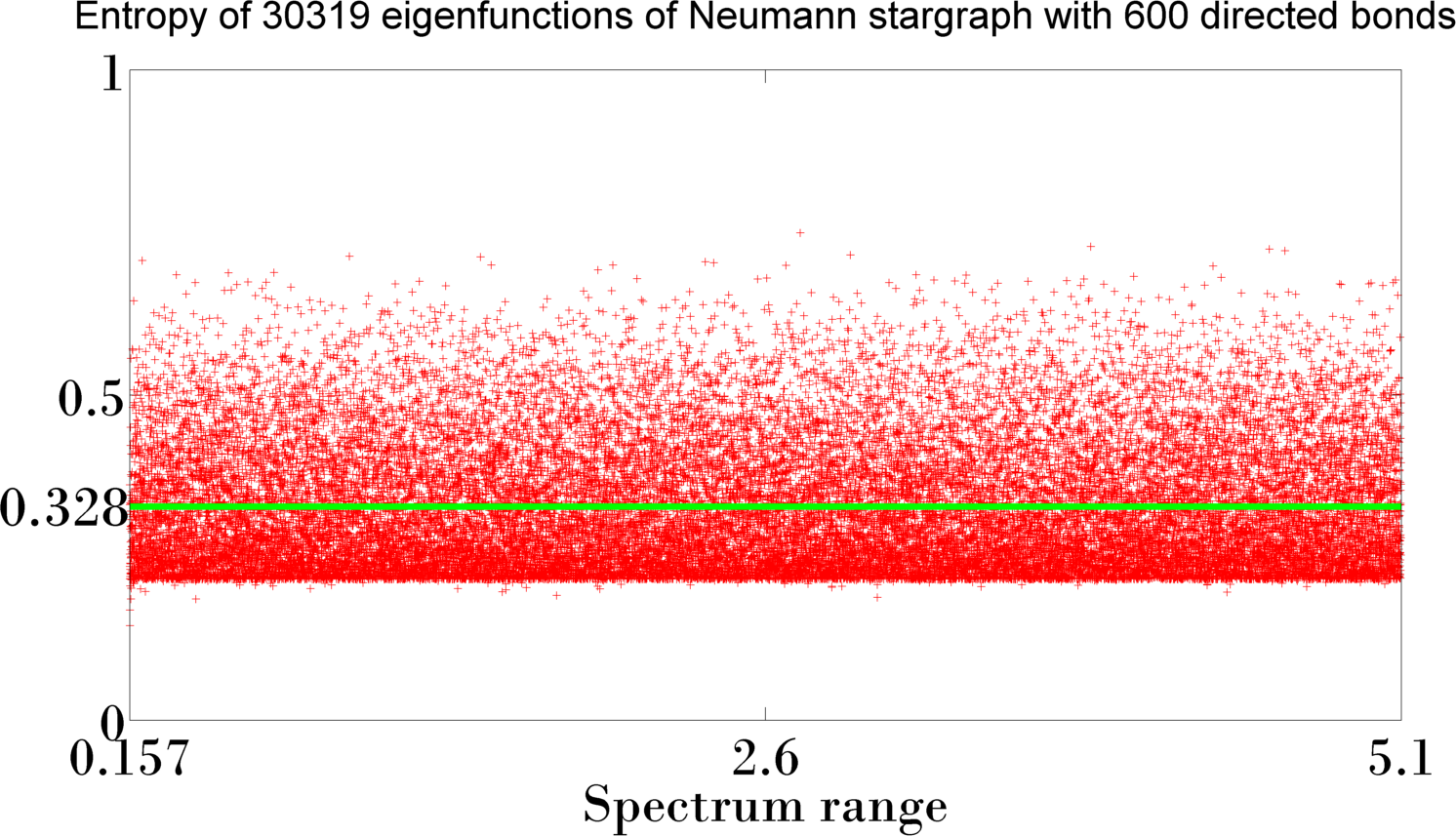}
\end{center}
\caption{Entropy of eigenfunctions for different graphs and different boundary conditions: 
The two plots on the top are for  a regular graph with degree 6 and 602 vertices which corresponds to 3612 directed bonds, 
top left we have equi-transmitting boundary conditions and top right Neumann boundary conditions on the vertices. 
The two plots on the bottom are for star graphs. Bottom left is with equi-transmitting boundary conditions and bottom right 
for Neumann boundary conditions. }
\label{fig:distrib_entropies}
\end{figure}

In this section we will compare our estimates on the entropy with numerical computations and discuss as well some connections to previous results 
on eigenfunctions on graphs, in particular \cite{SchaKot03,CdV14} and \cite{KeaMarWin03,BerKeaWin04}. 

To model expanders we choose random $d$-regular graphs, with $d=6$. For large size such graphs will be with high probability be expanders, see \cite{HooLinWig06}, 
but they do not necessarily have large girth. But they have with high probability very few short cycles \cite{HooLinWig06}, so are quite close to graphs with large girth. 
For the equi-transmitting boundary condition we choose the local $S$-matrix given by \eqref{eq:equi_trans_S} with \eqref{eq:Leg} and for 
Neumann we have the matrix in   \eqref{eq:Neumann_S}. The length we choose randomly from the interval $[2,10]$

In Figure \ref{fig:distrib_entropies} we show the entropies of all eigenfunctions in a certain spectral range for a $d$-regular and a star graph with 
both choices of boundary conditions. We see that for equi-transmitting boundary conditions the entropies are very large for both graphs, 
and have a very narrow distribution, indicating that the eigenfunctions behave very uniform. In particular the values for the entropy are 
well above the lower bound of $1/2$ we derived for expanders with large girth. For Neumann boundary condition on the regular graph the 
entropies a large to, but not quite as large as in the  equi-transmitting case, and the distribution is a bit wider as well and we see a few outlier, i.e., eigenfunctions 
with a rather small entropy. Finally the Neumann star graph shows a very different behaviour, the entropies have a much broader distribution and are much smaller. 

We will discuss now in some more detail how the entropy varies with the size of the graph. 

\begin{figure}[t]
\begin{center}
\includegraphics[width=7cm, height=4cm]{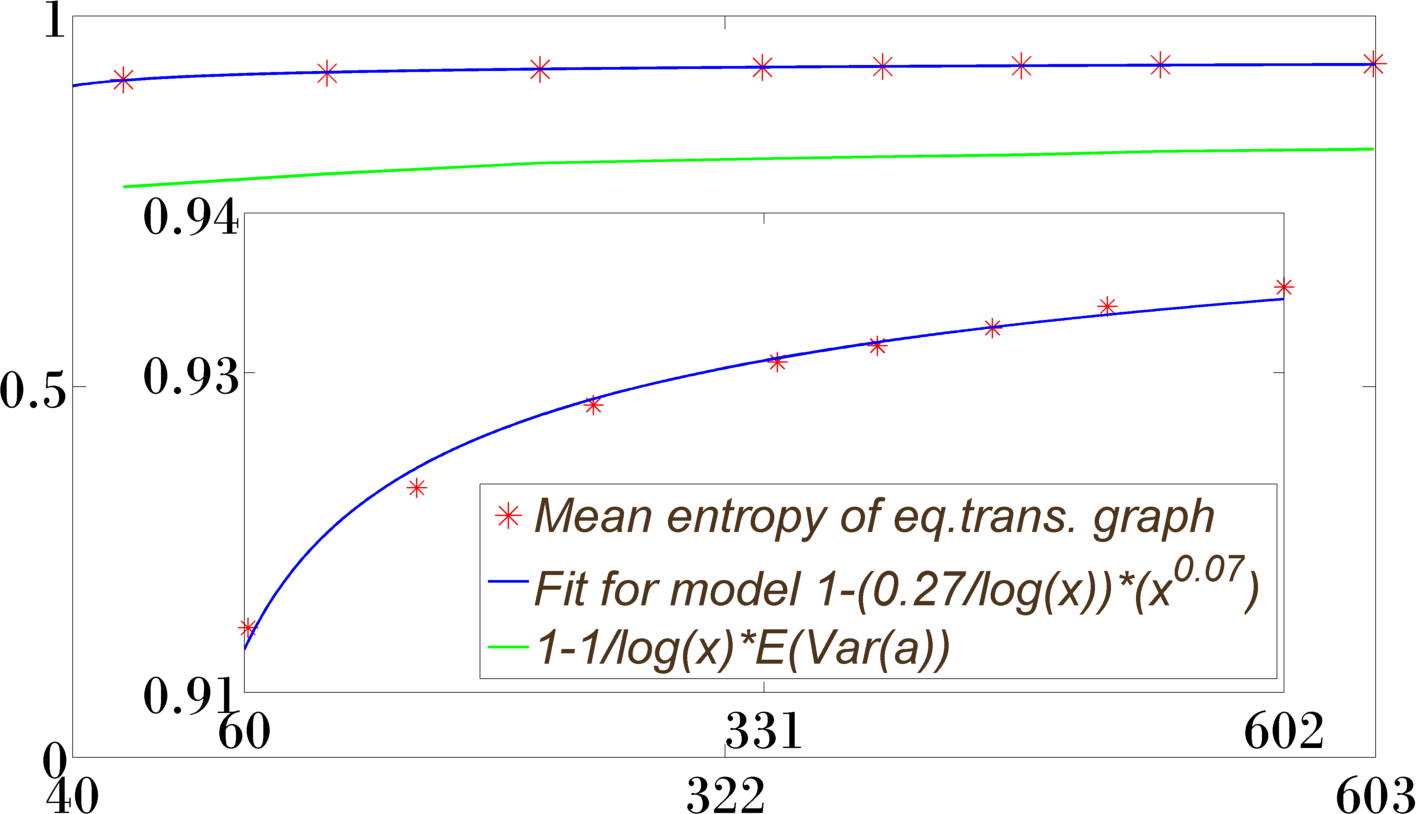}
\includegraphics[width=7cm, height=4cm]{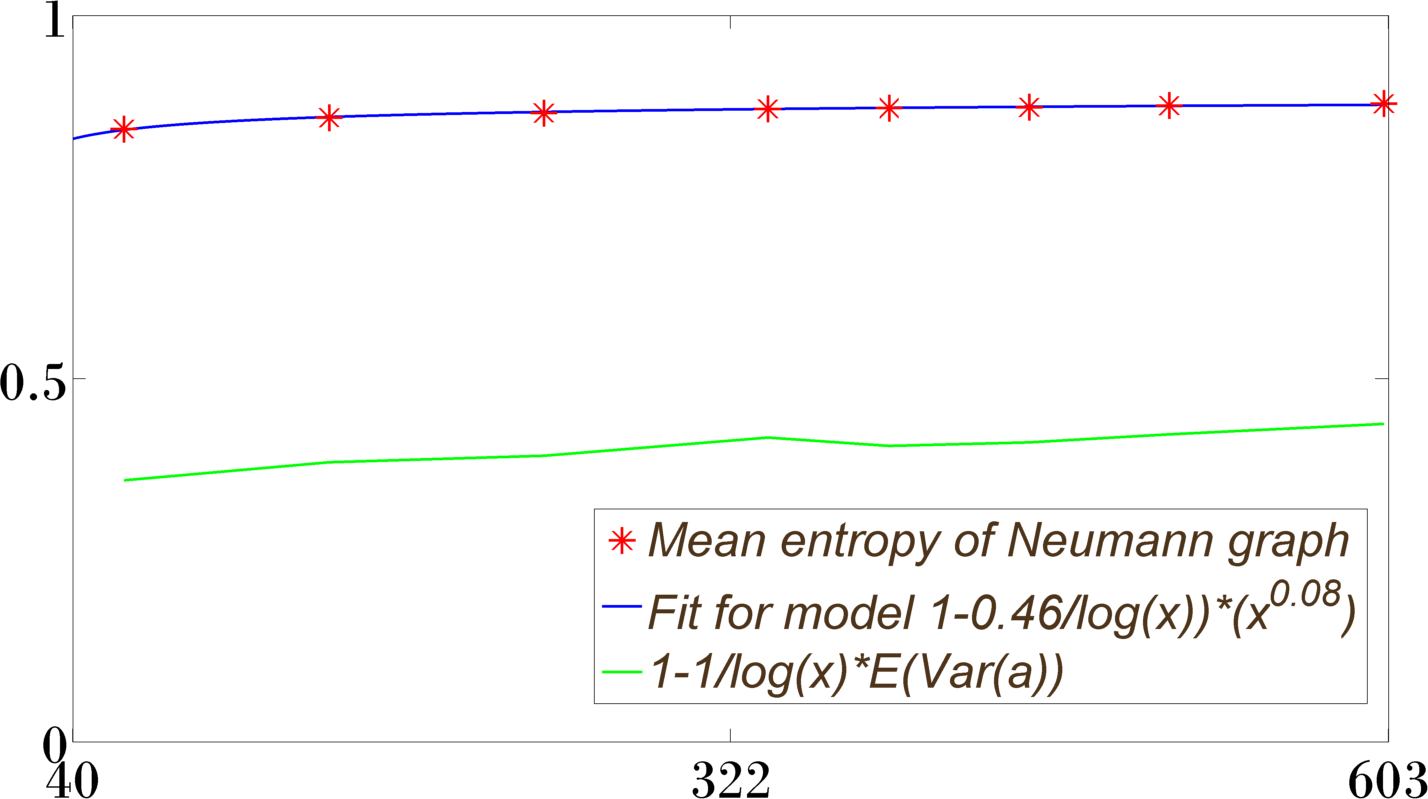}
\end{center}
\begin{center}
\includegraphics[width=7cm, height=4cm]{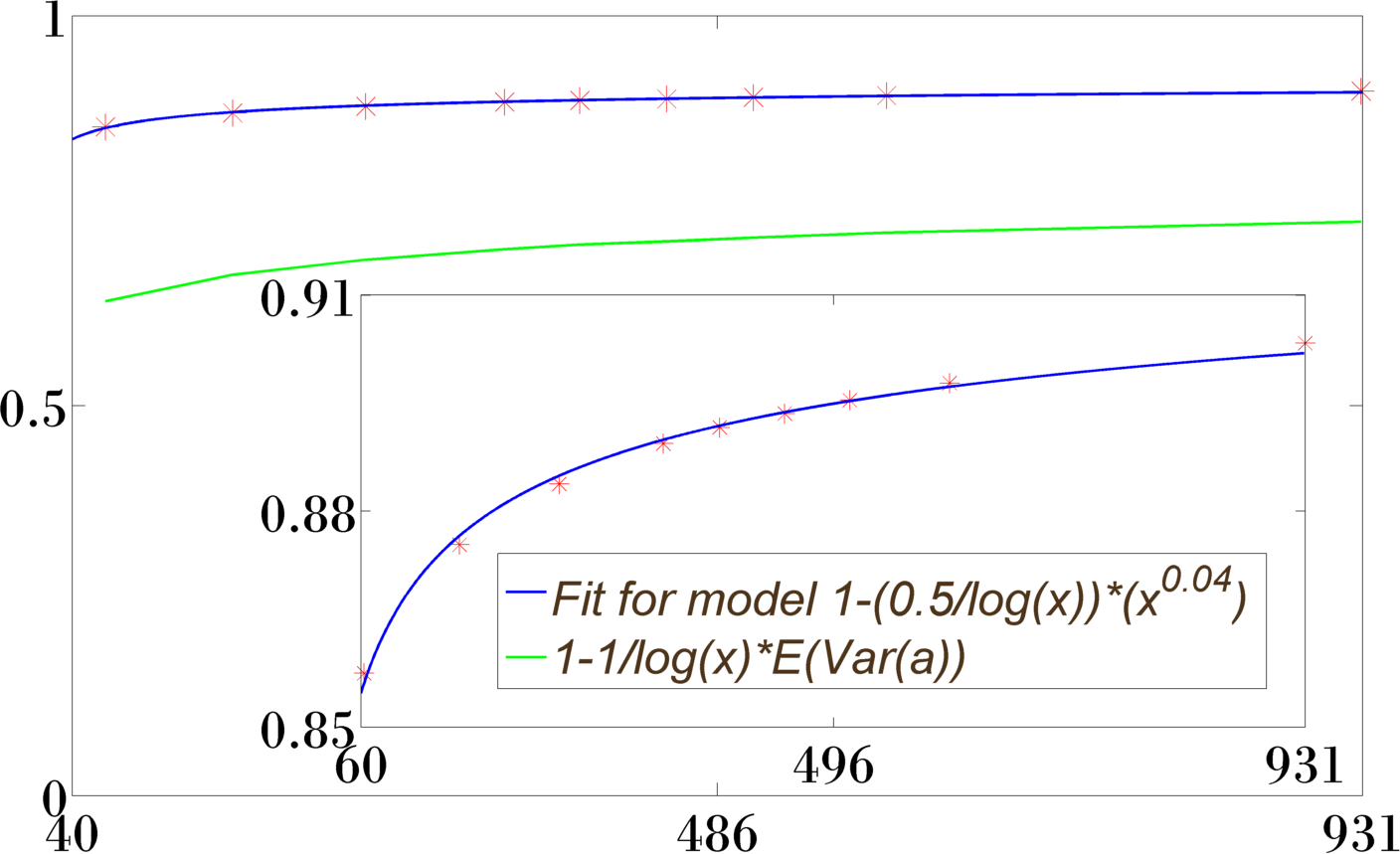}
\includegraphics[width=7cm, height=4cm]{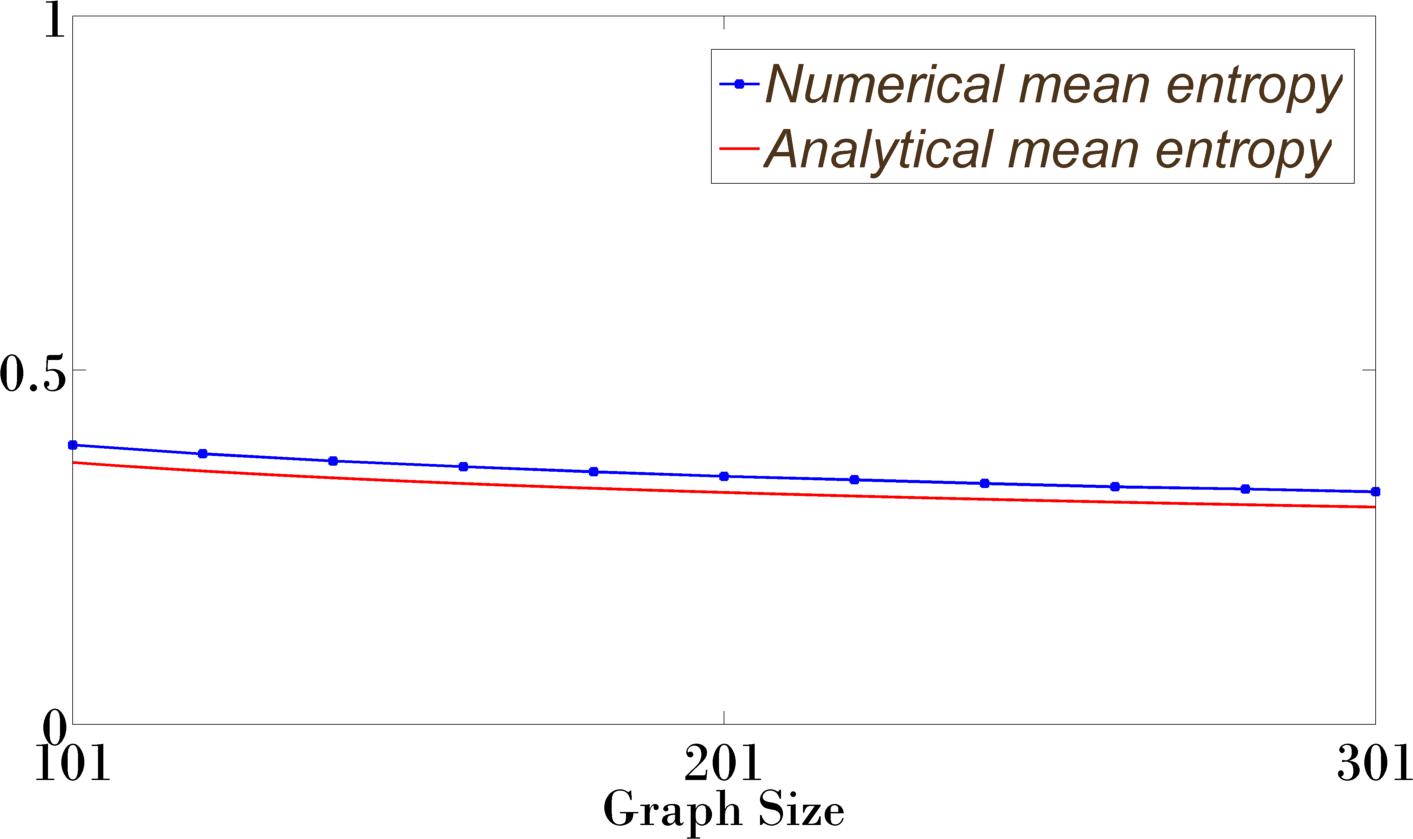}
\end{center}
\caption{Mean entropy of eigenfunctions for different graphs and different boundary conditions, averaged over a spectral window: 
The two plots on the top are for  a regular graph with degree 6 and 602 vertices which corresponds to 3612 directed bonds, 
top left we have equi-transmitting boundary conditions and top right Neumann boundary conditions on the vertices. 
The two plots on the bottom are for star graphs. Bottom left is with equi-transmitting boundary conditions and bottom right 
for Neumann boundary conditions. The green curve is the lower bound obtained from Lemma \ref{lem:entr_var} using numerical data for the variance.}
 \label{fig:growth_entropies}
\end{figure}

\subsection{Relation to quantum ergodicity and the variance}

The plots in Figure \ref{fig:growth_entropies} for the graphs with equi-transmitting boundary conditions all show an increase of the 
entropy of eigenfunctions with size of the graph. Furthermore for the same graph Figure \ref{fig:distrib_entropies} showed that the 
distribution of the entropies are narrowly concentrated around the mean. So it looks as if the entropy of eigenfunctions for these graphs 
approaches the maximal value for large graphs, and the eigenfunctions become equidistributed.

The test this further we will look at another common quantity to measure how equidistributed a vector is,  the variance. 
If $\ba\in \C^B$ with $\norm\ba=1$, then the vector is equidistributed if 
$\abs{a_b}^2\approx 1/B$, or $B\abs{a_b}^2\approx 1$, $b=1, \cdots , B$. The 
variance then measures how far the components of the vector deviate on average  from being equidistributed,  
\begin{equation}
V(\ba):=\frac{1}{B}\sum_{b=1}^B (B\abs{a_b}^2-1)^2\,\, ,\quad \text{where}\quad \norm{\ba}=1\,\, .
\end{equation}
If the variance is small then the vector is close to equidistribution. 

The variance can be used to estimate the entropy: 

\begin{lem}\label{lem:entr_var}
Let $\ba\in \C^B$ and $\norm{\ba}=1$, then 
\begin{equation}
S_N(\ba)\geq 1-\frac{1}{\ln B}\, V(\ba)\,\, .
\end{equation}
\end{lem}

\begin{proof}
This is consequence of the basic inequality $\ln(1+x)\leq x$, for $x>-1$. We have 
\begin{equation}
\begin{split}
\ln \abs{a_b}^2&=\ln\bigg( \frac{1}{B}(1+(B\abs{a_b}^2-1))\bigg)\\
&=-\ln B+\ln (1+(B\abs{a_b}^2-1))\leq -\ln B+(B\abs{a_b}^2-1)\,\, , 
\end{split}
\end{equation}
and inserting this into the definition of $S_N(\ba)$ gives immediately the result. 
\end{proof}

The variance is closely related to quantum ergodicity and the random wave model for eigenfunctions graphs which was developed and 
studied in \cite{GnuKeaPio10}. Let $D$ be a diagonal $B\times B$ matrix, with diagonal matrix elements $D_b$, and consider 
\begin{equation}
\bar D:=\frac{1}{\sum_{b=1}^BL_b} \sum_{b=1}^BL_b D_b\,\, ,\quad F_D:=\lim_{N\to\infty} \frac{1}{N} \sum_{n=1}^N\frac{\abs{\la \ba(n), DL\ba(n)\ra}^2}{\la \ba(n), L\ba(n)\ra} 
\end{equation}
where $L$ is the diagonal matrix of bond-length. Then one of the main results derived in \cite{GnuKeaPio10} is that if $\hat G_n$ are a family of graphs with finite spectral gap 
then if $D^{(n)}$ is a family of diagonal matrices whose elements are bounded uniformly in $n$ and which satisfies $\bar D^{(n)}=0$, then  there exist a $C>0$, 
independent of $n$, such that 
\begin{equation}\label{eq:QE}
F_{D^{(n)}} \leq \frac{C}{B}\,\, .
\end{equation}
This is a quantum ergodicity statement with an optimal rate. 

To connect this to the variance, let us choose $D$ such that $L_bD_b\in\{\pm 1\}$ are independently distributed with equal probability for $+1$ and $-1$, then 
$\E(\abs{\la \ba(n), DL\ba(n)\ra}^2=\sum_{b\in \hat E} \abs{a_b(n)}^4$ and so \eqref{eq:QE} implies that on average 
\begin{equation}
\sum_{b=1}^B \abs{a_b(n)}^4=O\bigg(\frac{1}{B}\bigg)\,\, .
\end{equation}
This implies that the variance of an eigenfunction satisfies on average
\begin{equation}
V(\ba(n))=O(1)\,\, .
\end{equation}
Notice that this is related to the inverse participation ratio, used for instance in \cite{SchaKot03}. We don't expect the variance to go to $0$, because that would imply equidistribution on microscopic scales, we rather expect that at that scale quantum fluctuations are present. 
Quantum ergodicity then predicts equidistribution on macroscopic scales where we average over many bonds. 

So if the average of the variances tend to a constant, then Lemma \ref{lem:entr_var} suggest that the entropy will tend at a  logarithmic rate to $1$. We computed the variances for 
the d-regular graphs and the star graph with Neumann boundary conditions,  they stay almost constant and show only a very mall increase with the size of the graph. 
 For comparison we included the lower bound from  Lemma \ref{lem:entr_var} with the numerically determined variances in the plot of the 
entropies in Figure \ref{fig:growth_entropies}. We fitted as well a model function of the form $f(B)=1-\alpha \frac{B^{\beta}}{\ln B}$ where a small $\beta\geq 0$ models the slight 
increase of the  averaged variances over the observed $B$ interval. We see that the model fits the date very well.

Let us now turn to the star graph with Neumann boundary conditions.  In Figure \ref{fig:growth_entropies} we see that the mean entropy decreases in a fashion which is compatible with the prediction in Theorem  \ref{thm:average_S}, but the numerically observed data are larger than the prediction. We do not know the reason for this deviation, it could be that 
the prediction in Theorem  \ref{thm:average_S} is only reached for very large graph size. Another issue is that the properties of the star graph are quite sensitive to the rational independence of the length of the edges, and this could pose a problem for numerical computations with a large number of edges.

\begin{figure}[t] 
\begin{center}
\includegraphics[width=8cm, height=5cm]{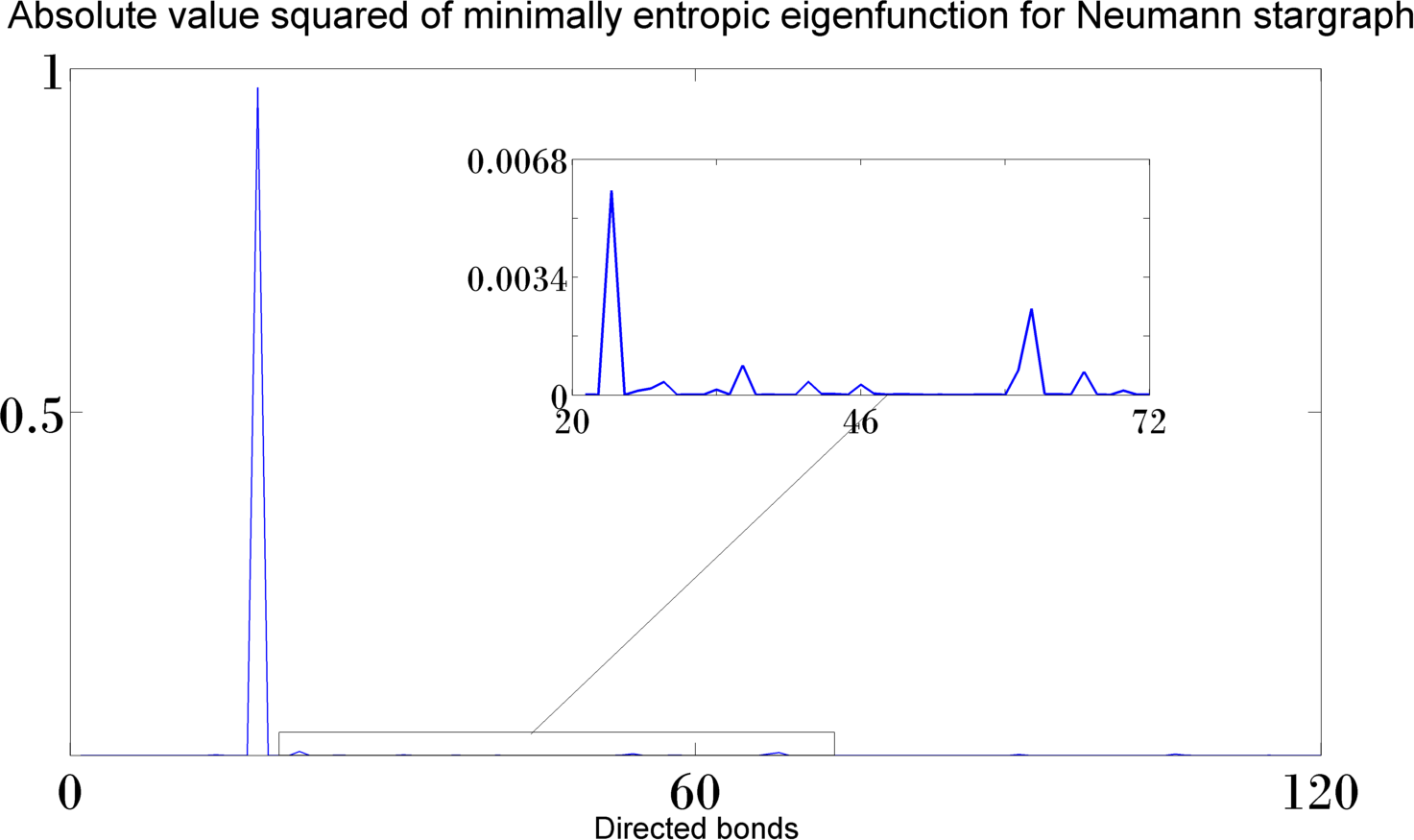}
\includegraphics[width=8cm, height=5cm]{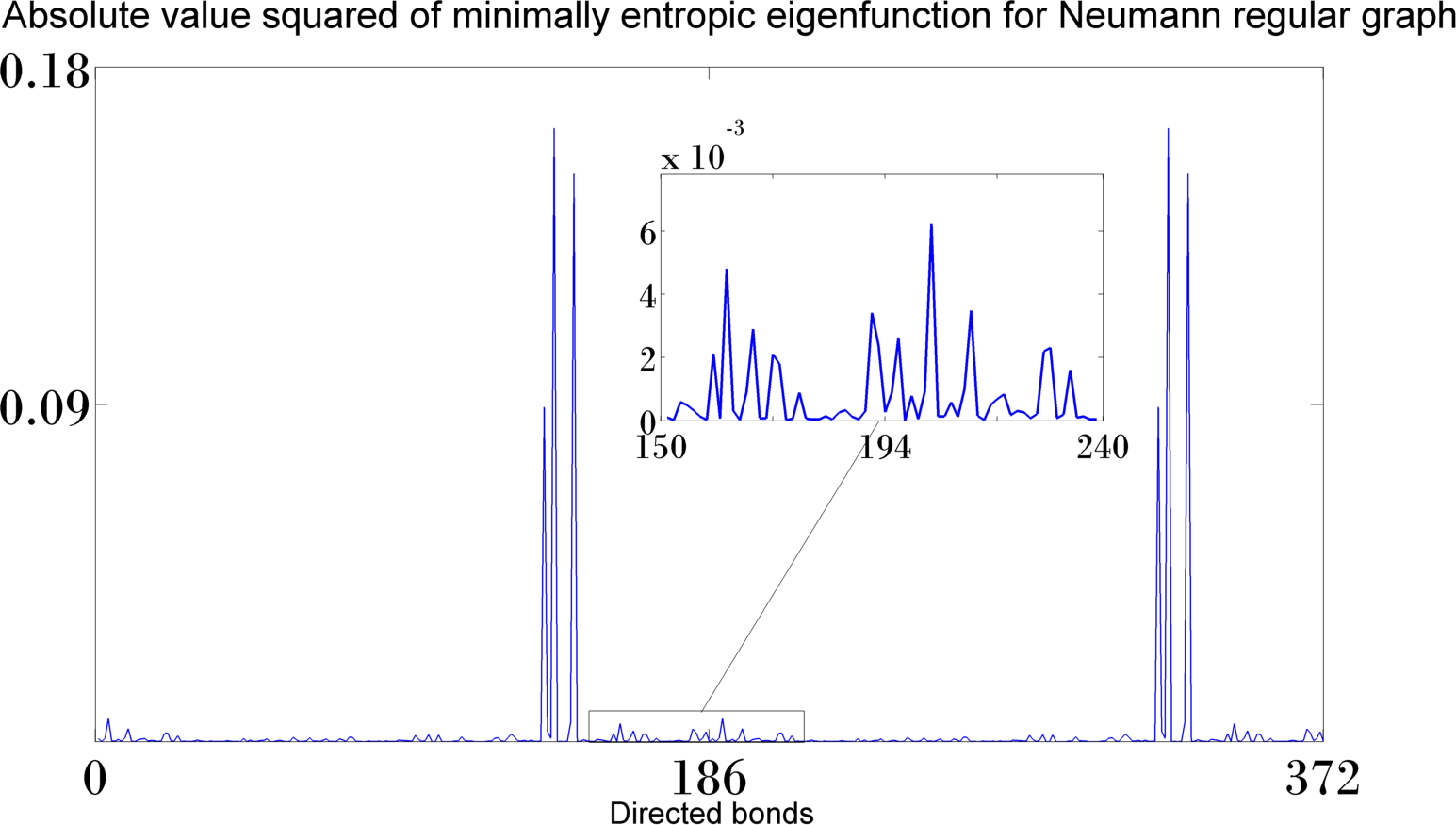}
\end{center}
\caption{Two examples of eigenfunctions with small entropy, shown are plots of the absolute value of the components of $\bA$ and $\ba$, respectively. On the left is an eigenfunction 
on a Neumann star graph with 120 edges, it is an eigenfunction with small energy ($k= 0.1579...$) which is almost completely concentrated on one edge only. 
On the right is an eigenfunction on a Neumann $6$-regular graph with $186$ edges, this eigenfunction is concentrated on a cycle of period $3$. }
\label{fig:localised}
\end{figure}

\subsection{Eigenfunctions with small entropy}

For the quantum graphs with Neumann boundary conditions we found in the numerical data some eigenfunctions with exceptionally small entropy, 
both on the regular graphs and on the star graphs.  It is well known that Neumann boundary conditions allow for eigenfunctions to concentrate 
on closed cycles, see  \cite{SchaKot03,CdV14}. Let us recall why this is the case: a function $\psi$ satisfies Neumann boundary conditions at a vertex $i$ 
if  
\begin{itemize}
\item[(a)] $\psi_{[ij]}(0)=\psi_{[ij']}(0)$ for all $j,j' $ with $i\sim j$ and $i\sim j'$
\item[(b)]    $\sum_{j\sim i}\psi_{[ij]}'(0)=0$. 
\end{itemize}
If the eigenfunction vanishes on some bonds connected to $i$, then by (a) 
$\psi_{[ij]}(0)=0$ for all bonds connected to $i$, so the  condition (b) can be satisfied if at least two of the terms in the sum are non-zero and cancel each other, 
and it can not be satisfied if only one term is non-zero. This way one can piece together eigenfunctions which are concentrated on a closed cycle, provided the 
length of the bonds in that cycle are rationally dependent. If they are not rationally dependent then one can   still find a sequence of $k_j$ such the corresponding 
eigenfunctions concentrate for large $j$ on the closed cycle. On a d-regular graph the shortest  cycles have period $3$ and they appear with finite probability 
in a random d-regular graph, therefore we expect to see some eigenfunctions which concentrate on them. In Figure \ref{fig:distrib_entropies} we see in the 
plot of the entropies of eigenfunctions on the 6-regular graph with Neumann boundary condition one eigenfunction with rather small entropy, this 
eigenfunction is plotted in the right panel of Figure \ref{fig:localised}. We see that it is highly localised on $3$ edges (corresponding to $6$ bonds), and inspection of the adjacency matrix show that these edges are adjacent and so form a 3-cycle.

On Neumann star graphs the shortest cycle on which eigenfunctions can concentrate for large $k$ has two edges, see \cite{BerKeaWin04}, and we 
see plenty of eigenfunctions of  this type in our numerical data. But surprisingly we see as well eigenfunctions which are almost completely concentrated on 
one edge only, see the left panel in Figure \ref{fig:localised}. The Neumann boundary conditions prohibit a function from being concentrated on one edge only, 
but the example we show belongs to a graph with a large number of edges, and although the eigenfunction is large on one edge, and small on all others, the 
large number of edges allow to compensate for the smallness of the eigenfunction on them. Notice that in  Figure \ref{fig:localised} we plot the modulus squared 
of the coefficients, which increases the perceived difference in the size of the coefficients. In the boundary conditions the coefficients themselves enter and 
the large number of small ones add up to cancel the one large one in condition (b). We notice as well that the eigenvalue of this eigenfunction is very small and that 
further eigenfunctions of this type all appeared at the bottom of the spectrum. Based on this observation we can get a heuristic explanation for the 
appearance of these eigenfunctions. 

Let $\sigma_k$ be the $S$-matrix \eqref{eq:sigma_k}, then the eigenvalues $k_n$ of the star graph are determined by the condition that 
$\sigma_k$ has an eigenvalue $1$, hence if we follow the eigenvalues of $\sigma_k$ on the unit circle as $k$ varies, we find an eigenvalue of the quantum graph whenever 
one of the eigenvalues of $\sigma_k$ crosses $1$. We will write the eigenvalues of $\sigma_k$ as $\ue^{\ui\theta_j(k)}$, $j=1, \cdots ,\abs{E}$, and we will follow their evolution for small $k$. 
The matrix $\sigma_0$ has an eigenvalue $1$ with multiplicity $1$ and an eigenvalue $-1$ with multiplicity $\abs{E}-1$. Now a standard identity in the spirit of the Feynman Hellman theorem gives 
\begin{equation}\label{eq:change_theta}
\frac{\ud \theta_j(k)}{\ud k}=2 \la \bA_j, L\bA_j\ra
\end{equation}
where $\bA_j$ is a normalised eigenvector of $\sigma_k$ with eigenvalue $\ue^{\ui\theta_j(k)}$, see \cite{BerWin10}. From this we learn that the eigenvalues 
$\ue^{\ui\theta_j(k)}$ move counterclockwise around the unit circle if we increase $k$, and in particular that there is a gap between 
$k_0=0$ and $k_1$ which is determined by the time it takes for the fastest eigenvalue starting at $\theta_j(0)=\pi$ to reach $\theta_j(k_1)=2\pi$. 
But \eqref{eq:change_theta} tells us that the way to make this gap, and therefore $k_1$, small, is to have an eigenvector $\bA_j$ which is 
concentrated on the longest edge, so that the right hand side of  \eqref{eq:change_theta} becomes as large as possible, i.e., $\la \bA_j, L\bA_j\ra=L_{max}$, 
and then 
\begin{equation}
k_1=\frac{\pi}{2 L_{max}}\,\, . 
\end{equation}
 Reversing the argument, we conclude that if we have a graph with one edge significantly longer than the others and 
if $k_1\approx \frac{\pi}{2 L_{max}}$, then the corresponding eigenfunction has to be concentrated on the longest edge. This is a phenomenon which 
can become more pronounced for large graphs, since the boundary conditions allow then for a larger concentration on a single bond.
The eigenfunction shown on the left panel of  Figure \ref{fig:localised} is on a graph with $L_{max}=9.9691$ and then we obtain 
$\frac{\pi}{2 L_{max}}=0.1576$ which is very close to the eigenvalue $k_1= 0.1579$. This confirms our heuristic picture of the mechanism behind 
the eigenfunctions localised almost completely on a single bond.

{\bf{Acknowledgements:}} This work was carried out using the computational facilities of the 
Advanced Computing Research Centre at the  University of Bristol.


\small

\end{document}